\documentclass[11pt]{article}

\usepackage[margin=1in]{geometry}

\usepackage[utf8]{inputenc}
\usepackage[T1]{fontenc}
\usepackage{lmodern}
\usepackage{mathpazo}

\usepackage{xspace}
\usepackage{authblk}

\usepackage{enumitem}
\usepackage[dvipsnames]{xcolor}
\definecolor{ptblue}{RGB}{15,76,129}
\definecolor{cobalt}{rgb}{0.0, 0.28, 0.67}
\usepackage{tikz}
\usetikzlibrary{arrows.meta, positioning, calc, decorations.pathreplacing, fit, backgrounds}

\usepackage{amsmath,amsfonts,amssymb,amsthm,mathtools}

\let\OLDforall\forall
\renewcommand{\forall}{\;\OLDforall\:}
\let\OLDexists\exists
\renewcommand{\exists}{\;\OLDexists\,}

\usepackage[square,sort]{natbib}
\usepackage{hyperref}
\hypersetup{
linktocpage,
colorlinks=true,
citecolor=cobalt,
urlcolor=ptblue,
linkcolor=Plum,
pdftitle={Best-of-Both-Worlds Fairness for Mixed Goods and Chores},
pdfauthor={Haris Aziz, Xiaolin Bu, Xinhang Lu, Simon Mackenzie, Mashbat Suzuki, Biaoshuai Tao, and Toby Walsh},
}
\usepackage{cleveref}
\usepackage{crossreftools}
\theoremstyle{plain}
\newtheorem{theorem}{Theorem}[section]

\newtheorem*{claim-non}{Claim}

\newtheorem{lemma}[theorem]{Lemma}
\newtheorem{proposition}[theorem]{Proposition}

\theoremstyle{definition}
\newtheorem{definition}[theorem]{Definition}

\theoremstyle{remark}

\makeatletter
\newcommand{\EF}[1]{\if\relax\detokenize\expandafter{\@firstofone#1{}}\relax \text{EF}\xspace\else \text{EF#1}\fi}
\makeatother
\newcommand{\EFOne}{\EF{1}\xspace}
\newcommand{\EFX}{\EF{X}\xspace}
\newcommand{\EFM}{EF{M}\xspace}
\newcommand{\SD}{\text{SD}\xspace}
\newcommand{\SDEF}{\text{SD-EF}\xspace}
\newcommand{\SDEFOne}{\text{SD-EF1}\xspace}
\newcommand{\SDEFTwo}{\text{SD-EF2}\xspace}

\newcommand{\bbR}{\mathbb{R}}
\newcommand{\exan}{ex-ante\xspace}
\newcommand{\expo}{ex-post\xspace}

\newcommand{\PP}{\mathbb{P}}
\newcommand{\EE}{\mathbb{E}}
\newcommand{\RR}{\mathbb{R}}

\title{Best-of-Both-Worlds Fairness for Mixed Goods and Chores}
\author[1]{Haris Aziz}
\author[2]{Xiaolin Bu}
\author[3]{Xinhang Lu}
\author[1]{Simon Mackenzie}
\author[1]{Mashbat Suzuki}
\author[2]{Biaoshuai~Tao}
\author[1]{Toby Walsh}
\affil[1]{UNSW Sydney, \nolinkurl{{haris.aziz, simon.w.mackenzie, mashbat.suzuki, t.walsh}@unsw.edu.au}}
\affil[2]{Shanghai Jiao Tong University, \nolinkurl{{lin_bu, bstao}@sjtu.edu.cn}}
\affil[3]{Kyushu University, \nolinkurl{xinhang.lu@inf.kyushu-u.ac.jp}}
\date{}

\begin{document}
\maketitle

\begin{abstract}
We study the fundamental problem of fairly dividing indivisible items among agents with additive utilities.
In our model, an item can be a good yielding non-negative utilities to some agents and simultaneously a chore yielding negative utilities to others.
We take the best-of-both-worlds perspective and our goal is to construct a randomized allocation that is exactly fair ex ante while also being supported on ex post approximately fair allocations.
The fairness notions examined in this paper are \emph{envy-freeness (EF)} and its well-known relaxation \emph{envy-freeness up to one item (EF1)}.
Our main result is that ex-ante EF and ex-post EF1 can be achieved simultaneously.
To achieve this, we introduce a novel probabilistic Hall-type matrix decomposition
that intricately correlates the fractional assignments of goods and chores. We resolve this
decomposition problem by combining continuous minimax duality---via Sion's minimax
theorem---with carefully designed biased flow networks.
\end{abstract}

\section{Introduction}

The \emph{fair} division of valuable scarce resources or undesirable burden among agents is a common problem that arises frequently in society, and has attracted extensive and ongoing interest in the fields of mathematics, economics, operations research, and computer science~\citep{BramsTa96,BrandtCoEn16,Moulin03,Moulin19,Rothe24}.
The quintessential notion of fairness is arguably \emph{envy-freeness (EF)}, which stipulates that each agent values their own allocated bundle weakly more than any other bundle in the allocation~\citep{Foley67}.
In other words, each agent does not envy any other agent.

In this paper, we are concerned with the fundamental problem of fairly dividing \emph{indivisible items}.
Due to its potential to capture a wide range of real-world applications, indivisible-items allocation has received significant attention in recent years, with focus mainly on allocating (purely) goods~\citep{AmanatidisAzBi23} or (purely) chores~\citep{GuoLiDe23}.
Note that \emph{goods} (resp., \emph{chores}) are items yielding non-negative (resp., negative) utilities to the agents.
In addition to theoretical research, the developed algorithms have already been deployed in real life and made accessible to the public through web applications; see, e.g., Course Match~\citep{BudishCaKe17}, Spliddit~\citep{GoldmanPr15} and Kajibuntan~\citep{IgarashiYo23}.

We focus on a more general model that allows agents to report more expressive utilities over the set of indivisible items to be divided.
More specifically, in our model of mixed goods and chores, an indivisible item can be a good to some agents and simultaneously a chore to other agents.
Use cases include household chore division between couples or tenants, the division of assets and liabilities between multiple parties, etc.
In our first example, an agent may want to avoid cooking, while another agent may find cooking is desirable and relaxing.
Similarly, in our second example, while some stakeholders seek risky assets as they have the potential to yield high returns, others may be risk-averse and would like to receive low-risk assets.

When items are indivisible, envy-freeness is too demanding to achieve, as envy-free allocations may not always exist.
One common way to circumvent this issue is to weaken envy-freeness.
For instance, \emph{envy-freeness up to one item (\EFOne)} is a very well-known relaxation of envy-freeness, which was first introduced in the goods setting~\citep{LiptonMaMo04,Budish11} and later generalized to the mixed goods and chores setting~\citep{AzizCaIg22}.
An allocation of indivisible items is said to be \EFOne if any envy an agent has towards another agent can be eliminated when we (hypothetically) remove a chore from the envious agent's bundle or a good from the envied agent's bundle.
With a mix of goods and chores, an \EFOne allocation always exists when agents have additive utilities~\citep{AzizCaIg22,BhaskarSrVa21}.

Another natural and common alternative to achieve fairness
is to use randomization, with which we can specify a probability distribution (i.e., \emph{lottery}) over a set of deterministic allocations.
Randomization has been used in a variety of contexts such as assignment problem~\citep{BogomolnaiaMo01}, apportionment~\citep{Grimmett04,GolzPePr26}, collective choice~\citep{BogomolnaiaMoSt05,BrandlBrPe21}, and resource allocation~\citep{BudishChKo13}.
With the power of randomization, it is easy to achieve \emph{ex ante} envy-freeness, e.g., via the following procedure: we choose a single agent uniformly at random and then allocate all of the items to the agent.
It can be seen that no agent envies another in expectation, although each realized deterministic allocation introduces a large amount of envy \emph{ex post}.

In the work of \citet{freeman2020best}, they focused on the goods setting, and initiated the study of constructing a randomized allocation (i.e., a probabilistic distribution over deterministic allocations) that is exactly fair ex ante (before the randomness is realized) and approximately fair ex post (after the randomness is realized).
This ``best of both worlds (BoBW)'' approach has attracted an active line of research.
Among other results, \citeauthor{freeman2020best} showed that for additive utilities, \exan \EF{} and \expo \EFOne are compatible.
Put differently, there always exists a randomized allocation that is envy-free in expectation and can be decomposed into a set of deterministic allocations, all of which satisfy \EFOne.
This motivates a natural question in our setting:\footnote{A more general question was stated as Open Question~4 in the survey of \citet{LiuLuSu24}.}
\begin{quote}
\itshape
With a set of mixed goods and chores and agents having additive utilities, can we achieve simultaneously \exan \EF{} and \expo \EFOne?
\end{quote}

\subsection{Our Result and Technical Overview}

We settle the above question in the affirmative.

\begin{theorem}
\label{thm:main}
Under the mixed goods and chores setting with additive valuations, \exan EF and \expo \EFOne are compatible.
\end{theorem}

As we will describe shortly in Section~\ref{sec:related}, a majority of prior fair division research on best-of-both-worlds fairness is concerned with only goods.
Moreover, with appropriate and straightforward adaptation, the \exan \EF{} and \expo \EFOne BoBW result for the goods setting continue to hold for the chores settings~\citep[see, pp.~1685 of][]{AzizFrSh24}.
However, showing \exan \EF{} is compatible with \expo \EFOne for mixed goods and chores becomes tricky and calls for novel conceptual ideas and mathematical techniques, which may be of independent interest.

In addition to our main result, in Appendix~\ref{sect:SD-BoBW}, we also study the setting where only ordinal preferences of agents are available, with the adapted fairness notions SD-EF and SD-\EFOne defined in terms of stochastic dominance.
We show that \exan SD-EF is compatible with \expo SD-EF2, but not SD-\EFOne.

In what follows, we first describe technical challenges, followed by our core technical innovation which helps tackle the challenges and finally a proof sketch for Theorem~\ref{thm:main}.

\subsubsection{Technical Challenges}
It is known that the Probabilistic Serial (PS) lottery and its variant, the recursive Probabilistic Serial lottery, can achieve both \exan EF and \expo \EFOne when the items are goods only.
In our setting with mixed goods and chores, both PS-lottery and recursive PS-lottery can be extended to handle the following two special cases:
\begin{enumerate}
\item Each item is valued non-negatively by at least one agent (we call such an item a \emph{subjective good}).

\item Each item is valued negatively by all agents (we call such an item an \emph{objective chore}).
\end{enumerate}
We will discuss these in detail in Section~\ref{sect:PS}.

Since every item is either a subjective good or an objective chore, a naive approach to solving the general case is to
independently run known best-of-both-worlds mechanisms (such as the Probabilistic Serial (PS) or recursive PS rule) for both the set of subjective goods and the set of objective chores.
However, this only yields ex-ante EF and ex-post EF2.\footnote{Loosely speaking, EF2 requires that any envy from one agent towards another agent can be eliminated by the hypothetical removal of two items.}

To improve the \expo guarantee to EF$1$, the lotteries for subjective goods and objective chores must be intricately correlated.
This is the main technical challenge for our problem.
To overcome this challenge, we introduce many mathematical techniques that have not been used in the fair division literature before.
Notably, we derive a novel probabilistic Hall-type matrix decomposition theorem.

\subsubsection{Core Technical Innovation: A Probabilistic Hall-Type Decomposition}
\label{sect:core_tech}

To resolve the technical challenge of correlating the lotteries for subjective goods and objective chores, we formulate and solve the following abstract mathematical problem in Theorem~\ref{thm:multi-T}.

\paragraph{The Probabilistic Hall-Type Decomposition Problem:}
Let $X \in [0, 1]^{n \times m}$ be a fractional bipartite matching matrix where each row sums to~$r/n$ (for some integer~$r < n$) and each column sums to at most~$1$.
Given a collection of disjoint sets $T_1, \dots, T_\ell \subseteq [n]$, does there exist a probability distribution~$\lambda$ over size-$r$ integral matchings that implements~$X$, such that for every set~$T_i$ and every subset~$J \subseteq T_i$, the set of matched rows~$K_\mu$ satisfies:
\begin{equation*}
\mathbb{P}_{\mu \sim \lambda}(K_{\mu} \cap T_i \neq \emptyset \text{ and } K_{\mu} \cap T_i \subseteq J) \le \frac{|J|}{|T_i|}?
\end{equation*}

This condition can be viewed as a probabilistic analog to Hall's Marriage Condition.
Intuitively, it prevents the matching from overly concentrating on any subset $J\subseteq T_i$.

Proving the existence of a single lottery~$\lambda$ that simultaneously satisfies this exponential family of constraints is highly non-trivial.
To overcome this, we introduce the following mathematical techniques that are largely distinct from the standard fair division toolkit.
We believe this mathematical framework will be of independent interest for tackling correlated decomposition problems in algorithmic mechanism design.

We begin by describing our techniques for handling the special case with a single set, i.e., with $\ell=1$, in Theorem~\ref{thm:fixed-T}.
In the analysis below, let $T$ be this single set.

\paragraph{Continuous Weighted Reformulation.}
We first lift the discrete subset constraints into a continuous domain.
Instead of reasoning about all subsets~$J$, we introduce a weight vector~$w \in \Delta(T)$ over the standard simplex and show that the subset constraints can be captured by the following expected-value inequality $$\mathbb{E}_{\mu\sim\lambda}[\varphi_w(K_\mu \cap T)] \le \frac{1}{|T|} \sum_{j \in T} w_j,$$ which must hold for all weight vectors~$w$.

\paragraph{Sion's Minimax Theorem.}
Even after lifting, finding a universal~$\lambda$ to satisfy the inequality for all weight vectors~$w$ simultaneously remains challenging.
This is equivalent to showing the minimax problem $\min_{\lambda} \max_{w} g(\lambda, w) \le 0$ for all~$w$ with
$$g(\lambda, w) \coloneqq \mathbb{E}_{\mu \sim \lambda}[\varphi_w(K_\mu \cap T)] - \frac{1}{|T|} \sum_{j \in T} w_j.$$
To address this, by invoking \emph{Sion's minimax theorem}, we swap the order of the optimization and obtain $\max_{w} \min_{\lambda} g(\lambda, w)$.
This quantifier swap is a crucial simplification.
Instead of finding a universal distribution~$\lambda$, it suffices to find a distribution~$\lambda_w$ tailored for each \emph{fixed} weight vector~$w$.

\paragraph{Biased Flow Networks.}
Finally, to construct~$\lambda_w$ for a fixed weight vector~$w$, we first sort the rows according to~$w$.
This sorting step is the \emph{only} point where the resulting distribution~$\lambda_w$ depends on the specific weight vector~$w$.
Based on this ordered structure, we build a specialized \emph{biased flow network} whose capacities are integral and encode the marginal constraints that each row sums to~$r/n$.
We route a feasible fractional flow through a sequence of chain nodes.
By applying the flow-integrality lemma, we decompose this flow into a probability distribution~$\lambda_w$ over integral matchings.
The structure of the network ensures that $g(\lambda_w, w) \le 0$.

\paragraph{Resolving the General Case with Multiple Sets.}
For multiple pairwise disjoint sets, the matchings are intrinsically coupled across different sets through a single probability distribution $\lambda$.
As a result, we cannot simply compute independent lotteries for each $T_i$.
To overcome this, we generalize our framework for the single-set case, which shares the same conceptual structure while involving more details that need to be carefully handled; full details are deferred to Appendix~\ref{append:proof-thm-multi-T}.

At a high level, we first expand the continuous weighted reformulation from a single simplex to a product space of $w = (w^1, \dots, w^\ell) \in \prod_{i=1}^\ell \Delta(T_i)$.
We aggregate the expected-value constraints $g^i(\lambda, w^i)$ for each $T_i$ into a global objective function $G(\lambda, w) = \sum_{i=1}^\ell g^i(\lambda, w^i)$, where $g^i$ corresponds to the constraint associated with $T_i$.
The summation preserves the desired properties of each $g^i$ to apply Sion's minimax theorem, which allows us to again fix an arbitrary choice of $w$ to find a tailored distribution $\lambda_w$.
In the final step, we exploit the pairwise disjointness of the sets by sorting the rows \emph{locally} within each $T_i$ according to $w^i$.
This enables us to construct a collection of \emph{parallel} chain-nodes gadgets, one for each set.
A distribution $\lambda_w$ is obtained through a similar analysis, which will simultaneously satisfy $g^i(\lambda_w, w^i)\le 0$ for each set.

\subsubsection{Proof Sketch for Theorem~\ref{thm:main}}

Our algorithmic framework accomplishes this in three parts: a bundling preprocessing phase, followed by two distinct algorithmic approaches depending on the number of remaining objective chores.

\paragraph{Bundling Preprocessing (Section~\ref{sect:preprocessing}).}
As mentioned before, items are partitioned into subjective goods and objective chores.
In the first step, we group items carefully such that each group is viewed as a single ``meta-item''.
At the end, we reduce the original instance to an instance with three crucial structural properties:
\begin{enumerate}
\item each subjective good has a non-empty \emph{interest set} of agents who value it non-negatively;

\item the interest sets are pairwise disjoint, meaning that each agent values at most one subjective good non-negatively (this implies no more than $n$ subjective goods); and

\item each subjective good is \emph{chore-maximal}, meaning that the addition of any objective chore renders its overall value strictly negative for all agents.
\end{enumerate}
We then divide our analysis based on the number of remaining objective chores.

\paragraph{The Hard Case: More than $n$ Objective Chores (Section~\ref{sect:hardcase}).}
When the number of objective chores strictly exceeds the number of agents~$n$, we explicitly define a target fractional allocation that is \exan EF.
Subjective goods are divided equally among the agents in their respective interest sets, while objective chores (padded with dummy items) are fractionally allocated using the recursive PS rule.

The primary technical hurdle is decomposing this combined fractional allocation into a lottery over integral allocations that strictly satisfy \expo EF$1$.
We introduce a synchronized rounding scheme: an agent can be allocated a subjective good only if she is simultaneously assigned an objective chore (rather than a dummy item) in the \emph{first round} of the recursive PS decomposition, or if no agent in their interest set receives a chore in that round.
Due to the chore-maximality property, the disutility of the assigned chore outweighs the value of the assigned subjective good in every agent's perspective.
This makes the \expo fairness analysis similar to the case of chores-only instances, and it ensures \expo EF$1$.

The remaining difficulty is to show that such a correlated decomposition exists while preserving the \exan marginal probability.
For this, we impose a probabilistic Hall-type condition on the support of the first-round objective chore matchings.
We prove that if this condition holds, then subjective goods can be successfully allocated to satisfy both \exan and \expo constraints via a max-flow min-cut argument.
Proving the existence of a first-round chore decomposition that satisfies this Hall-type condition is the most technical part of our paper, which is stated in Section~\ref{sect:core_tech}.

\paragraph{The Easy Case: At Most $n$ Objective Chores (Section~\ref{sect:easycase}).}
When there are at most $n$ objective chores, we apply a distinct randomized combinatorial algorithm.

We first execute a secondary preprocessing phase to ensure that the subjective goods are \emph{good-minimal}, which guarantees that any \expo EF1 condition satisfied by removing a bundled-subjective good is inherently satisfied by removing a subjective good from the original instance, and the instance satisfies a \emph{small-goods-only} property, i.e., the disutility of any objective chore strictly outweighs the total utility of all subjective goods an agent values non-negatively.
Next, we design a three-step randomized lottery:
\begin{enumerate}
\item We uniformly at random sample a subset of agents (equal to the number of chores) and assign exactly one objective chore to each.
\item We run a serial dictatorship among these selected agents to distribute any subjective goods they value non-negatively.
\item Finally, the remaining agents (who received no chores) divide the leftover subjective goods via a standard PS-lottery padded with dummy items.
\end{enumerate}

Ex-post EF$1$ follows naturally from the small-goods-only and good-minimality property.
Ex-ante EF is established by employing a delicate \emph{coupling argument} and pairing symmetric outcomes of the lottery, such as swapping the roles or selection orders of two agents.

\subsection{Related Work}
\label{sec:related}

The literature of fair division has been developed rapidly~\citep[see, e.g., the surveys of][]{AmanatidisAzBi23,GuoLiDe23,LiuLuSu24,NguyenRo23,Suksompong21,Suksompong25}.
We will present here the most relevant papers to our work.
We first summarize the development of fairly allocating mixed goods and chores, followed by the best-of-both-worlds fairness results.

\subsubsection*{Fair Division of Mixed Goods and Chores}

\noindent
\textbf{Mixed \emph{Indivisible} Goods and Chores.}
An \EFOne allocation always exists~\citep{AzizCaIg22,BhaskarSrVa21}.
The problem becomes trickier if we ask for both \EFOne and the economic efficiency notion \emph{Pareto optimality (PO)}, and it remains open to date.\footnote{For additive utilities, \EFOne and PO allocations always exist if we allocate purely goods~\citep{CaragiannisKuMo19} or purely chores~\citep{Mahara26}.}
By introducing other envy-freeness relaxations, \citet{BarmanHVSe25} and \citet{BarmanVe26} showed that their proposed envy-freeness relaxations are compatible with PO.

When weakening \EFOne to \emph{proportionality up to one item (PROP1)}, for additive utilities, a PROP1 and PO allocation can be computed in strongly polynomial time, even if agents have asymmetric weights~\citep{AzizMoSa20}.

The existence and computation of other fairness notions, such as \emph{envy-freeness up to any item (EFX)}, \emph{maximin share (MMS) fairness} and relaxations of \emph{equitability}, as well as their compatibility with PO, have also been studied~\citep[e.g.,][]{HosseiniSiVa23,HosseiniSe25,KulkarniMeTa21,KulkarniMeTa21-AAAI,LivanosMeMu22}.
A more detailed account of these works can be found in the overview of \citet[Section~4]{LiuLuSu24}.

\medskip
\noindent
\textbf{Mixed \emph{Divisible} Goods and Chores (i.e., Mixed Manna).}
\citet{BogomolnaiaMoSa17} initiated the study of dividing a mix of (homogeneous) \emph{divisible} goods and chores (i.e., mixed manna) and focused on an economic concept called \emph{competitive equilibrium}.
They established the existence of competitive equilibria.
Follow-up research has then devoted to studying the computation of an equilibrium~\citep{ChaudhuryGaMc23,GargMcHo23}.

\medskip
\noindent
\textbf{Mixed Indivisible Goods and Chores plus a Cake.}
Most recently, \citet{AzizLuMa26} studied a more general model in which a heterogeneous divisible cake and a set of mixed indivisible goods and chores are divided among agents who have additive utilities over the resources.
They investigated an envy-freeness relaxation called \emph{envy-freeness for mixed resources (\EFM)}, which combines \EF{} and \EFOne in a natural way.
\citeauthor{AzizLuMa26} showed that an \EFM allocation always exists in the aforementioned setting.

\subsubsection*{Best-of-Both-Worlds (BoBW) Fairness}

In the literature of fair division, BoBW fairness has been mostly studied for the goods setting, with only a few exceptions.

\citet{SunCh25} studied the mechanism design problem of indivisible-items allocation.
By restricting agents' utility functions, they proposed randomized mechanisms that are strategyproof in expectation as well as fair and efficient both ex ante and ex post, for chores setting and for the setting with mixed goods and chores, respectively.

The literature we surveyed below is mostly concerned with goods and we will make it clear when chores setting is studied.
For additive utilities, \exan \EF{} and \expo \EFOne are compatible~\citep{AzizFrSh24}.
Furthermore, \citeauthor{AzizFrSh24} showed several impossibility results if we want to achieve both fairness and economic efficiency; however, the compatibility between \exan \EF{} (or even proportionality), \expo \EFOne and \expo PO remains open to date.
Nevertheless, by restricting agents' utilities to be \emph{bi-valued}, \citet{BuLiLi24} showed that there always exists a randomized allocation that is \exan \EF{}, \expo \EFX and \expo \emph{fractionally Pareto optimal (fPO)}.\footnote{For \expo deterministic allocations, fractionally Pareto optimality (fPO) is a stronger notion than PO.}

In a similar vein, when restricted to \emph{lexicographic} preferences, \citet{KavithaPaVa25} showed the compatibility between \exan \EF{}, \expo \EFOne and \expo PO as well as the compatibility between \exan $\frac{9}{10}$-\EF{}, \expo \EFX and \expo PO.

When assuming that agents' utility functions are more general than additive utilities, \exan $\frac{1}{2}$-EF and \expo EFX-with-charity are compatible for \emph{monotone} utilities~\citep{KavithaPaVa25}.
For \emph{subadditive} utilities, we have the following two sets of compatibility results: \exan $\frac{1}{2}$-\EF{}, \expo $\frac{1}{2}$-\EFX and \expo \EFOne~\citep{FeldmanMaNa24}, as well as \exan $\frac{1}{2}$-proportionality and \expo EFX-with-bounded-charity are compatible~\citep{KavithaPaVa25}.

\citet{HoeferScVa24} and \citet{AzizGaMi23} studied the setting where agents have entitlements and showed that unlike the equal entitlement setting, weighted envy-freeness (WEF) is even incompatible with weak WEF1.
On the positive side, \exan WEF, \expo WEF$(1, 1)$, and \expo weighted proportionality up to one good (WPROP1) are compatible~\citep{HoeferScVa24}.
Analogous results are known for chores: \exan WEF is compatible either with \expo WEF$(1, 1)$~\citep{WuZhZh25} or with \expo WPROP1~\citep{HVNi24}.

Other fairness notions, such as proportionality, maximin share, and equitability, have also been explored to achieve best-of-both-worlds fair division of indivisible goods~\citep{AkramiGaSh24,AkramiMeSe23,BabaioffEzFe22,BhaskarHVSe26}.

\citet{BuLiLi24} studied a model concerning a mix of \emph{divisible and indivisible} goods.
For bi-valued utilities, \citeauthor{BuLiLi24} showed that \exan proportionality is compatible with \expo \EFM.

In addition to resource allocation, best-of-both-worlds fairness has also been investigated in the context of committee voting~\citep{AzizLuSu23,Peters25,SuzukiVo24}, participatory budgeting~\citep{AzizLuSu24}, etc.

\section{Preliminaries}
For any positive integer~$k$, let $[k] \coloneqq \{1, 2, \dots, k\}$.
Denote by $N = [n]$ the set of $n$ agents and $M = [m]$ the set of $m$ indivisible items.
Each agent $i\in N$ has a valuation function $v_i:\{0,1\}^M\rightarrow \mathbb{R}$, where $v_i(\emptyset)=0$.
Throughout the paper, we assume each agent's valuation function is \emph{additive}, that is, $v_i(S)=\sum_{h\in S}v_i(\{h\})$ for $S\subseteq M$.
For simplicity, we use $v_i(h)$ to denote $v_i(\{h\})$.

Under the setting with both goods and chores, we classify item $h$ as an \emph{objective good} if $v_i(h)\ge 0$ for all agents $i\in N$, and an \emph{objective chore} if $v_i(h)<0$ for all agents $i\in N$.
We classify item $h$ as a \emph{subjective good} if $v_i(h)\ge0$ for at least one agent $i\in N$.
Note that with a slight abuse of terminology, this definition of subjective goods includes all objective goods.
In the following part of our paper, we partition the item set $M$ into two disjoint sets $G$ and $Z$, where $G$ contains all the subjective goods and $Z$ contains all the objective chores.
We denote a subjective good by $g\in G$ and an objective chore by $c\in Z$.

\paragraph{Allocation.}
A \emph{fractional allocation} is specified by a non-negative $n \times m$ matrix $X = (X_{i h})_{i \in N, h \in M}$, where $X_{i h} \in [0, 1]$ denotes the fraction of item~$h$ being allocated to agent~$i$, and $\sum_{i \in N} X_{i h} = 1$ for each item~$h \in M$.
Each row $X_i = (X_{i h})_{h\in M}$ represents the fractional allocation of agent~$i$.
Each agent $i$'s utility for $X_i$ is naturally extended as $v_i(X_i)=\sum_{h\in M}X_{ih}\cdot v_i(h)$.

An \emph{integral allocation} is a fractional allocation $X = (X_{i h})_{i \in N, h \in M}$ with $X_{i h} \in \{0, 1\}$ for all~$i \in N$ and $h \in M$.
For simplicity, we denote an integral allocation by $A=(A_1,\ldots,A_n)$, where $A$ is a partition of $M$ and $A_i$ is the bundle assigned to agent $i$.

A \emph{randomized allocation}, also called a \emph{lottery} of integral allocations, is a probability distribution over integral allocations $\{(p_k, A^k)\}_{k \in [K]}$, where, for every $k \in [K]$, $A^k$ is an integral allocation implemented with probability $p_k \in [0, 1]$, and $\sum_{k \in [K]} p_k = 1$.
We say that randomized allocation $\{(p_k, A^k)\}_{k \in [K]}$ implements a fractional allocation~$X$, or $\{(p_k, A^k)\}_{k \in [K]}$ is a decomposition of $X$, if for each agent $i\in N$ and each item $h\in M$, $X_{ih}$ is the marginal probability that $h$ is allocated to agent $i$, that is, $X_{ih} = \sum_{k \in [K]} p_k \cdot \mathbb{I}[h \in A^k_i]$.

A randomized allocation is said to satisfy some \emph{\exan} property if the fractional allocation it implements satisfies this property, and is said to satisfy some \emph{\expo} property if every integral allocation in its support satisfies this property.

\paragraph{Fairness Notion.}
One of the most natural fairness notions is envy-freeness, which requires that no agent envies any other agent.
\begin{definition}
    A fractional allocation $X$ satisfies \emph{envy-freeness (EF)}, if $v_i(X_i)\ge v_i(X_j)$ for any pair of agents $i$ and $j$.
\end{definition}

For integral allocations, envy-freeness cannot always be guaranteed.
We therefore consider its standard relaxation, which allows envy to be eliminated by the removal of one item.

\begin{definition}
    An integral allocation $A$ satisfies \emph{envy-freeness up to one item (EF$1$)}, if for any pair of agents $i$ and $j$, either of the following holds:
    \begin{itemize}
        \item $i$ does not envy $j$: $v_i(A_i)\geq v_i(A_j)$, or
        \item there exists $h\in A_i\cup A_j$ such that $v_i(A_i\setminus\{h\})\geq v_i(A_j\setminus \{h\})$.
    \end{itemize}
\end{definition}
Alternatively, an integral allocation $A$ is EF1 if, for any pair of agents $i$ and $j$, either of the following holds:
\begin{itemize}
    \item $v_i(A_i) \geq v_i(A_j \setminus \{h\}) $ for some $h\in A_j$ with $v_i(h)\geq 0$, or
    \item $v_i(A_i \setminus \{h\}) \geq v_i(A_j) $ for some $h\in A_i$ with $v_i(h)< 0$.
\end{itemize}

\subsection{Probabilistic Serial Rule}
\label{sect:PS}
The probabilistic serial (PS) rule, introduced by~\citet{BogomolnaiaMo01}, gives randomized allocations that are \exan EF.
Its recursive variant was applied by~\citet{freeman2020best} to show the compatibility of \exan EF and \expo EF1.
\citet{aziz2020simultaneously} later showed that the original PS rule can also be used to guarantee both \exan EF and \expo EF1.\footnote{The two papers \citep{freeman2020best} and \citep{aziz2020simultaneously} have been combined to a single journal paper \citep{AzizFrSh24}.}

\paragraph{PS rule.}
Below, we describe the probabilistic serial (PS) rule by~\citet{BogomolnaiaMo01} in the context of \emph{divisible} objective goods.
Starting at time $t=0$, each agent simultaneously ``eats'' their favorite available item at a unit speed.
If multiple agents prefer the same item, it is consumed collectively at a rate equal to the number of agents eating it.
Once an item is fully depleted, the agents consuming it immediately move to their next most-preferred available item.
This process continues until all items are consumed, by which time $t=m/n$.

It is easy to see that the fractional allocation output by the PS rule is envy-free.

\paragraph{PS-lottery.}
The above-mentioned PS rule yields a fractional allocation that is envy-free.
When the objective goods are indivisible, \citet{AzizFrSh24} show that this fractional allocation can be decomposed into a probability distribution over integral EF$1$ allocations, and thus obtain a best-of-both-worlds fairness for objective goods.
To see this, we first assume without loss of generality that $m=nT$ for $T\in \mathbb{Z}_{\ge 0}$ (so that the PS rule ends in $T$ units of time), for otherwise dummy items with value $0$ to each agent are added to $M$.
For each $i\in N$, $h\in M$, and $t\in[T]$, let $x_{i,t,h}$ be the fraction of item $h$ eaten by agent $i$ in the time interval $[t-1,t]$ under the PS rule and $M_{i,t}=\{h:x_{i,t,h}>0\}$ be the set of items (possibly partially) consumed by agent $i$ in the time interval $[t-1,t]$.
Consider the polytope in $\mathbb{R}^{nTm}=\mathbb{R}^{m^2}$ defined by the following two sets of constraints:
\begin{equation}\label{eqn:PS1}
    \forall i\in N, t\in[T]: \sum_{h\in M_{i,t}}x_{i,t,h}=1
\end{equation}
\begin{equation}\label{eqn:PS2}
    \forall h\in M: \sum_{i=1}^n\sum_{t=1}^Tx_{i,t,h}=1
\end{equation}
The fractional allocation output by the PS rule satisfies the two sets of constraints above, and is thus in the polytope.
Moreover, the coefficient matrix of this polytope is totally unimodular, so each vertex corresponds to an integral allocation $A=(A_1,\ldots,A_n)$ where each agent $i$ gets exactly one item from $M_{i,t}$ for each $t\in[T]$ (due to constraint (\ref{eqn:PS1}) and the integrality of the variables).
It then suffices to show that each vertex gives an EF$1$ allocation.
For each vertex, we call the unique item in $M_{i,t}$ allocated to agent $i$ \emph{the item allocated to agent $i$ at the $t$-th round}, and name this item $h_{i,t}$.
Then, for any two agents $i,j$ and any $t_1<t_2$, we have $v_i(h_{i,t_1})\ge v_i(h_{j,t_2})$ since agent $i$ values each item in $M_{i,t_1}$ weakly higher than each item in $M_{j,t_2}$ due to the property of PS (that items with higher values are consumed first).
Thus, agent $i$ will not envy agent $j$ if $h_{j,1}$ were removed from agent $j$'s bundle.

\paragraph{Recursive PS-lottery.}
Prior to \citet{AzizFrSh24}, \citet{freeman2020best} introduces the \emph{recursive PS-lottery}, which also admits an interpretation as a simultaneously ``eating'' process at a unit speed.
In contrast to the PS-lottery, recursive PS proceeds for $T$ discrete rounds (again, we assume $T\in\mathbb{Z}_{\geq0}$).
In each round, each agent will consume their most-preferred available item \emph{for one unit of time}, resulting in a partial fractional allocation.
This partial fractional allocation will be immediately decomposed into a lottery of partial integral allocations where each agent receives exactly one item from the set of items she has (possibly partially) consumed in this round.
When multiple lotteries exist, we may choose any of them.
Under each partial integral allocation, the process is then recursively executed for another unit of time among the unallocated items.
The process terminates after all items are allocated, i.e., after $T$ rounds.
In contrast to the PS-lottery where we first complete the PS rule and then decompose the obtained fractional allocation, the recursive PS-lottery alternately performs the eating (PS for one unit of time) and the decomposition steps in every unit of time.

The recursive PS-lottery is also \exan EF and \expo EF$1$.
It is \exan EF since EF holds for the partial fractional allocation in each time unit.
By the nature of the recursive PS-lottery, for each agent, items with higher values are consumed first.
Therefore, the \expo EF$1$ property holds for the same reason as it is in PS-lottery.

For both PS-lottery and recursive PS-lottery, we have shown the following proposition (which will be used later) for each realized integral allocation.
\begin{proposition}[\citet{AzizFrSh24}, \citet{freeman2020best}]
\label{prop:ps-lottery}
    Fix an integral allocation in the PS-lottery or the recursive PS-lottery.
    Let the item received by agent $i$ in round $t$ be $h_{i,t}$.
    Then, for any two agents $i,j$ and any two rounds $t_1<t_2$, we have $v_i(h_{i,t_1})\ge v_i(h_{j,t_2})$.
\end{proposition}

\paragraph{Extensions beyond objective goods.}
Both of the above rules can be extended to the setting with only subjective goods or with only objective chores to get \exan EF and \expo EF$1$ guarantee.
In an instance that only contains subjective goods where each item is non-negatively valued by at least one agent, we can apply either rule after adding a sufficiently large number of dummy items, each of which has value $0$ to each agent, to $M$, so that each agent will not consume any item with negative value.
The \exan EF property holds for the same reason, and each integral allocation is \expo EF$1$ as envy from agent $i$ to agent $j$ will be eliminated by removing the item received in the first round in the integral allocation by $j$.

In an instance that only contains objective chores, applying either rule results in an \exan EF lottery for the same reason, and \expo EF$1$ is satisfied as envy from agent $i$ to $j$ will be eliminated by removing the item received in the last round by $i$.

For our setting with both subjective goods $G$ and objective chores $Z$, each rule satisfies \exan EF and \expo EF$2$.
Here, EF$2$ is a weaker notion than EF$1$, which allows envy to be eliminated by the removal of up to two items.
In particular, we can execute PS or recursive PS on $G$ and $Z$ separately as mentioned above, and then combine the lotteries.
Denote the fractional allocations implemented on $G$ and $Z$ as $X^G$ and $X^Z$ respectively, then the combined fractional allocation is naturally defined by setting $X_{i h} = X^G_{i h}$ for $h \in G$ and $X_{i h} = X^Z_{i h}$ for $h \in Z$.
The integral allocation can be arbitrarily combined while maintaining the marginal probability (e.g., taking their marginal product).
As both lotteries for $G$ and $Z$ are \exan EF and \expo EF$1$, the combined lottery is therefore \exan EF and \expo EF$2$.

\section{Bundling Preprocessing}
\label{sect:preprocessing}
Given an instance with mixed goods and chores $G\cup Z$, we first execute a bundling preprocessing phase to reduce the instance into a more desirable one.
This is done by first iteratively applying update rule (i), followed by update rule (ii), until no update is available.
\begin{itemize}
    \item[(i)] While there exists an agent $i\in N$ who values more than one subjective good in $G$ non-negatively, let $G_i\subseteq G$ be the set such that $G_i=\{g\in G:v_i(g)\ge 0\}$.
    We consider the items in $G_i$ as a single packed subjective good $g'$, that is, we remove $G_i$ from $G$ and add $g'$ to $G$;

    \item[(ii)] While $Z\neq\emptyset$, if there exists an agent $i\in N$, a subjective good $g\in G$ (possibly a packed item resulting from update rule (i)), and an objective chore $c\in Z$ such that $v_i(\{g,c\})\ge 0$, we consider $\{g,c\}$ as a single packed subjective good $g'$, that is, we remove $g$ and $c$ respectively from $G$ and $Z$, and add $g'$ to $G$.
\end{itemize}

In the resulting instance, we have a set $\{M_1,\ldots,M_\ell\}$ of subjective goods, and we still use $G$ to denote this set.

The set of objective chores, still denoted by $Z$, is a subset of that in the original instance.
Assume that $|G\cup Z|=m$.
For each subjective good $M_j$, define its \emph{interest set} $T_j$ as the set of agents who value it non-negatively, that is, $T_j=\{i\in N:v_i(M_j)\ge 0\}$.
It is straightforward to see that the new instance satisfies Proposition~\ref{prop:bundling}.

\begin{proposition}\label{prop:bundling}
    The bundling preprocessing phase outputs a set of subjective goods $G=\{M_1,\ldots,M_\ell\}$ and a set of objective chores $Z$, which have the following properties:
    \begin{enumerate}
        \item for any subjective good $M_j\in G$, its interest set is non-empty, i.e., $T_j\neq\emptyset$;
        \item each agent $i\in N$ non-negatively values at most one subjective good, i.e., for any pair of subjective goods $M_{j_1},M_{j_2}\in G$, their interest sets are disjoint, i.e., $T_{j_1}\cap T_{j_2}=\emptyset$;
        \item each subjective good $M_j\in G$ is \emph{chore-maximal}, that is, $v_i(M_j\cup \{c\})<0$ for each agent $i\in N$ and $c\in Z$.
    \end{enumerate}
\end{proposition}
\begin{proof}
    Property 1 is invariant under each of the two update rules. For Property 2, if an agent values two subjective goods non-negatively, we should continue to apply rule (i). If Property 3 is violated, we should continue to apply rule (ii).
\end{proof}

In the following, when given an instance with mixed goods and chores, we assume that it has been preprocessed as described above.
We will discuss two cases:
\begin{itemize}
    \item The hard case: there are more than $n$ objective chores after preprocessing (Section~\ref{sect:hardcase});
    \item The easy case: there are at most $n$ objective chores after preprocessing (Section~\ref{sect:easycase}).
\end{itemize}

\section{Hard Case: More than \texorpdfstring{$n$}{n} Objective Chores}
\label{sect:hardcase}
In this section, we assume there are more than $n$ objective chores after the preprocessing phase described in Section~\ref{sect:preprocessing}.

To construct a randomized allocation that is \exan EF and \expo EF$1$, we first define the target fractional allocation.
We separately consider the fractional allocation for subjective goods $G$ and objective chores $Z$, denoted by $X^G$ and $X^Z$ respectively.
Recall that the interest set of a subjective good $M_j\in G$ is defined as $T_j=\{i\in N:v_i(M_j)\ge 0\}$.
\begin{itemize}
    \item[$X^G$:] For each subjective good $M_j\in G$, let $M_j$ be divided among $T_j$ equally.
    Equivalently, each agent in $T_j$ receives a fraction of $\frac{1}{|T_j|}$ of $M_j$, and each agent outside $T_j$ receives no fraction of $M_j$.
    \item[$X^Z$:] Suppose the set of the objective chores has size $|Z|=kn+r$ for some integers $k\ge 1$ and $0\le r< n$.
    We introduce a set $D$ of $n-r$ dummy items where $v_i(d)=0$ for each agent $i\in N$ and $d\in D$.
    We then apply the recursive PS procedure on the combined set $Z\cup D$ and obtain a fractional allocation $X^{Z\cup D}$ where each agent receives a total fraction of $k+1$ items.
    Finally, $X^Z$ is obtained by removing the fractional allocation of the dummy items from $X^{Z\cup D}$.
\end{itemize}

The following lemma shows that the fractional allocation satisfies \exan envy-freeness.

\begin{lemma}\label{lem:exan-ef}
    A randomized allocation is \exan EF if it implements the fractional allocation defined above.
\end{lemma}
\begin{proof}
    For any two agents $i$ and $j$, the recursive PS guarantees $v_i(X^Z_i)\ge v_i(X^Z_j)$.
    Moreover, $v_i(X^G_i)=v_i(X^G_j)$ if $i$ and $j$ belong to the same interest set, and $v_i(X^G_i)\ge v_i(X^G_j)$ otherwise.
\end{proof}

We then need to define a randomized allocation rule that implements the fractional allocation rule described above.
This randomized allocation rule consists of two parts \emph{that are dependent}: a lottery over the fractional allocation $X^G$ of $G$ and a lottery over the fractional allocation $X^Z$ of $Z$.
\begin{itemize}
    \item Lottery over $X^Z$: To obtain \expo EF$1$ guarantee, for objective chores, we restrict our attention to the integral allocations obtained by decomposing $X^{Z\cup D}$ according to the recursive PS rule, and removing all the dummy items from each integral allocation.
    Noticing that during the first unit of time of the recursive PS, all agents will first consume all $n-r$ dummy items, and then a total fraction of $r$ objective chores.
    Therefore, in any integral allocation corresponding to this first unit of time, exactly $r$ objective chores will be assigned to the agents.
    We refer to this allocation of the first $r$ objective chores as the first-round allocation.
    We will carefully design the allocation rule \emph{for the first round}, while the allocation rule for the remaining rounds can be arbitrarily (as long as it follows the recursive PS rule).
    \item Lottery over $X^G$: For each subjective good $M_j\in G$, we will carefully design an allocation rule such that each agent in $T_j$ receives $M_j$ with probability $\frac1{|T_j|}$. This allocation rule is dependent on the first-round allocation rule for the lottery over $X^Z$.
\end{itemize}

Lemma~\ref{lem:expo-ef1} below gives a sufficient condition on the first-round recursive PS decomposition of $X^Z$ and the lottery over $X^G$ such that \expo EF$1$ is guaranteed.

\begin{lemma}\label{lem:expo-ef1}
    A randomized allocation is \expo EF$1$ if for any integral allocation in its support, the objective chores are allocated according to the above-mentioned decomposition; additionally, each subjective good $M_j\in G$ is allocated to an agent $i\in T_j$ only if either of the following holds:
    \begin{itemize}
        \item agent $i$ is assigned an objective chore from $Z$ (rather than a dummy item from $D$) in the first-round allocation, or
        \item no agent in $T_j$ is assigned any objective chore from $Z$ in the first-round allocation.
    \end{itemize}
\end{lemma}
\begin{proof}
    We consider the integral allocation $(A_1,\ldots,A_n)$ before removing the dummy items.
    EF$1$ of these allocations implies EF$1$ after the dummy items are removed.
    In each such integral allocation, each agent receives $k+1$ items in $Z\cup D$.
    For each $i\in N$ and $t\in[k+1]$, let $c_{i,t}$ be the corresponding chore of agent $i$ received in the $t$-th round during the recursive PS.
    Due to Proposition~\ref{prop:ps-lottery} (which also holds for recursive PS), we have $0\geq v_i(c_{i,t_1})\ge v_i(c_{j,t_2})$ for any two agents $i$ and $j$ and any two rounds $t_1< t_2$.

    Consider any pair of agents $i$ and $j$.
    If agent $j$ receives no subjective good, or agents $i$ and $j$ belong to different interest sets, EF$1$ is satisfied due to
    $$v_i\left(A_i\setminus\{c_{i,k+1}\}\right)\ge \sum_{t=1}^k v_i(c_{i,t})\ge \sum_{t=2}^{k+1} v_i(c_{j,t})\ge v_i(A_j).$$
    Otherwise, assume $i$ and $j$ belong to the same interest set $T_o$ and the corresponding subjective good $M_o$ is allocated to agent $j$.
    If $c_{j,1}$ is an objective chore, we have
    $$v_i\left(A_i\setminus\{c_{i,k+1}\}\right)= \sum_{t=1}^k v_i(c_{i,t})\ge \sum_{t=2}^{k+1} v_i(c_{j,t})\ge v_i(M_o\cup\{c_{j,1}\})+ \sum_{t=2}^{k+1} v_i(c_{j,t})= v_i(A_j),$$
    where the second inequality holds due to chore-maximality of each subjective good after bundling.
    If $c_{j,1}$ is a dummy item, according to our constraint, we know that $c_{i,1}$ is also a dummy item.
    Then we have
    $$v_i\left(A_i\setminus\{c_{i,k+1}\}\right)= \sum_{t=2}^k v_i(c_{i,t})\ge \sum_{t=3}^{k+1} v_i(c_{j,t})\ge v_i(M_o\cup\{c_{j,2}\})+ \sum_{t=3}^{k+1} v_i(c_{j,t})= v_i(A_j).$$
    We have shown that EF$1$ is satisfied in all cases.
\end{proof}

Theorem~\ref{thm:main} is immediately implied by the existence of such a randomized allocation due to the two lemmas above.
It remains to show that the fractional allocation in Lemma~\ref{lem:exan-ef} is decomposable into integral allocations satisfying the conditions of Lemma~\ref{lem:expo-ef1}.
In particular, we will show that there exists a specific lottery over the first-round allocations that implements the fractional allocation obtained after executing recursive PS for the first unit of time, so that we can properly allocate the subjective goods to satisfy the marginal probability and the constraint specified in the two lemmas above.

The following lemma formalizes this property of the lottery over the first-round allocations, which introduces a key concept of \emph{Hall-type condition}.

Recall that each first-round allocation will assign $r$ objective chores to $r$ agents, and thus can be viewed as a matching of size $r$.
A matching $\mu$ will alternatively be represented by an indicator matrix $\mathbf{1}_{\mu}\in\{0,1\}^{n\times m}$, where $(\mathbf{1}_{\mu})_{i,j}=1$ if item $j$ is matched with agent $i$ under $\mu$, and $(\mathbf{1}_{\mu})_{i,j}=0$ otherwise.
\begin{lemma}\label{lem:halltypecondition}
    Let \(X\in[0,1]^{n\times m}\) be the fractional allocation obtained after executing recursive PS for one unit of time.
    Suppose $X$ admits a decomposition
    \[
    X=\sum_{\mu}\lambda_{\mu}\,\mathbf{1}_{\mu}
    \]
    into size-\(r\) matchings such that, for each interest set \(T_j\) of \(M_j\) (where the interest sets are pairwise disjoint), and every subset \(J\subseteq T_j\), the following \emph{Hall-type condition} in~\Cref{eqn:hall} holds:
    \begin{equation}\label{eqn:hall}
        \mathbb{P}_{\mu\sim\lambda}\!\left(K_{\mu}\cap T_j\neq\emptyset \text{ and } K_{\mu}\cap T_j\subseteq J\right)\le \frac{|J|}{|T_j|},
    \end{equation}
    where \(K_\mu\) denotes the random variable corresponding to the set of agents that are matched under \(\mu\).
    Then, the subjective goods can be allocated to satisfy both the marginal probability in Lemma~\ref{lem:exan-ef} and the constraint in Lemma~\ref{lem:expo-ef1}.
\end{lemma}
\begin{proof}
    Fix a subjective good $M_j\in G$, and let $T_j=\{1,\ldots,t\}$.
    Suppose the decomposition has support $\{\mu_1,\ldots,\mu_{t'}\}$, where each size-$r$ matching $\mu_{i'}$ (corresponding to a first-round integral allocation) occurs with probability $\lambda_{i'}$.

    We formulate the problem as an $s$-$s'$ network flow problem.
    The flow network is constructed with two additional sets of vertices besides $s$ and $s'$.
    The first set contains $t$ vertices, each representing an agent in $T_j$.
    With a slight abuse of notation, we denote the vertices by $\{1,\ldots,t\}$.
    There is an edge with capacity $\frac{1}{t}$ from $s$ to each $i\in [t]$.
    The second set contains $t'$ vertices, each representing a matching in the support.
    With a slight abuse of notation, we denote the vertices by $\{\mu_1,\ldots,\mu_{t'}\}$.
    There is an edge with capacity $\lambda_{i'}$ from each matching $\mu_{i'}$ to $s'$.
    Moreover, there is an edge with capacity $\infty$ from $i\in[t]$ to $\mu_{i'}$ for $i'\in[t']$, if assigning $M_j$ to agent $i$ under matching $\mu_{i'}$ does not violate the constraint in Lemma~\ref{lem:expo-ef1}.
    Specifically, the edge occurs if agent $i$ is matched with some objective chore under $\mu_{i'}$, or all agents in $T_j$ remain unmatched under $\mu_{i'}$.
    Figure~\ref{fig:flow1} illustrates the network flow instance.

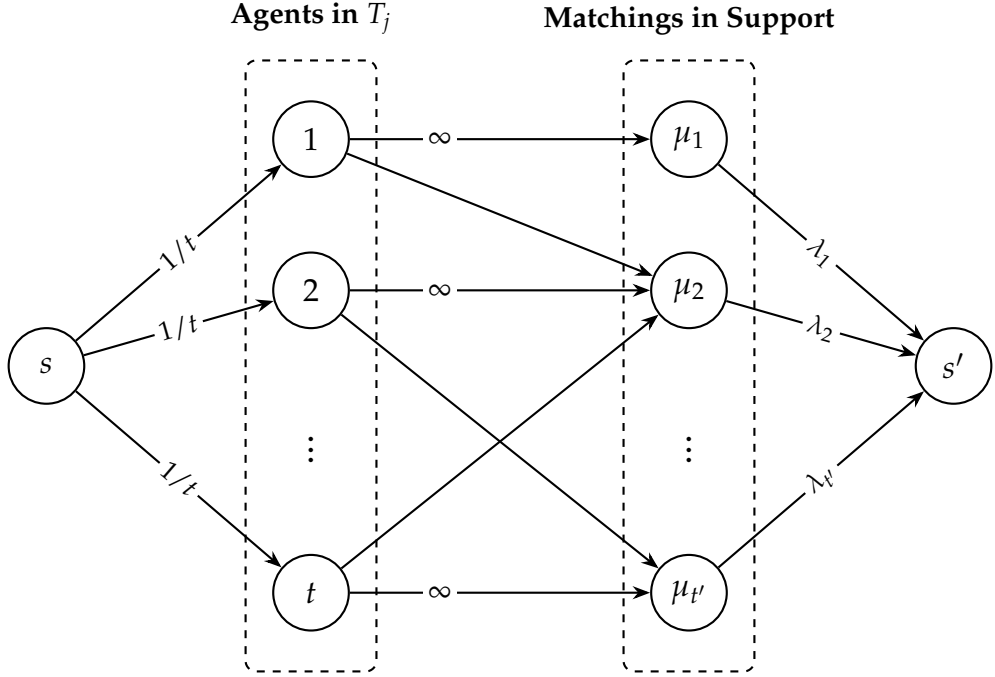
\begin{figure}
    \centering
\begin{tikzpicture}[
    >=Stealth,
    vertex/.style={circle, draw, minimum size=1cm, thick, inner sep=0pt, font=\large},
    flow/.style={fill=white, inner sep=2pt, font=\small},
    edge/.style={->, thick}
]

    \node[vertex] (s) at (0, 0) {$s$};

    \node[vertex] (a1) at (3.5, 3) {$1$};
    \node[vertex] (a2) at (3.5, 1) {$2$};
    \node (adots) at (3.5, -1) {\Large $\vdots$};
    \node[vertex] (at) at (3.5, -3) {$t$};

    \node[vertex] (m1) at (8.5, 3) {$\mu_1$};
    \node[vertex] (m2) at (8.5, 1) {$\mu_2$};
    \node (mdots) at (8.5, -1) {\Large $\vdots$};
    \node[vertex] (mt) at (8.5, -3) {$\mu_{t'}$};

    \node[vertex] (sprime) at (12, 0) {$s'$};

    \node[draw, thick, dashed, rounded corners, fit=(a1) (at), inner xsep=10pt, inner ysep=15pt] (agents_box) {};
    \node[above=0.2cm of agents_box, font=\bfseries] {Agents in $T_j$};

    \node[draw, thick, dashed, rounded corners, fit=(m1) (mt), inner xsep=10pt, inner ysep=15pt] (matchings_box) {};
    \node[above=0.2cm of matchings_box, font=\bfseries] {Matchings in Support};

    \draw[edge] (s) -- (a1) node[flow, pos=0.5, sloped] {$1/t$};
    \draw[edge] (s) -- (a2) node[flow, pos=0.5, sloped] {$1/t$};
    \draw[edge] (s) -- (at) node[flow, pos=0.5, sloped] {$1/t$};

    \draw[edge] (a1) -- (m1) node[flow, pos=0.3] {$\infty$};
    \draw[edge] (a1) -- (m2);

    \draw[edge] (a2) -- (m2) node[flow, pos=0.3] {$\infty$};
    \draw[edge] (a2) -- (mt);
    \draw[edge] (at) -- (m2);
    \draw[edge] (at) -- (mt) node[flow, pos=0.3] {$\infty$};

    \draw[edge] (m1) -- (sprime) node[flow, pos=0.5, sloped] {$\lambda_1$};
    \draw[edge] (m2) -- (sprime) node[flow, pos=0.5, sloped] {$\lambda_2$};
    \draw[edge] (mt) -- (sprime) node[flow, pos=0.5, sloped] {$\lambda_{t'}$};

\end{tikzpicture}
    \caption{The flow network for the proof of Lemma~\ref{lem:halltypecondition}. An intermediate edge $(i,\mu_{i'})$ represents that assigning $M_j$ to agent $i$ under matching $\mu_{i'}$ does not violate the constraint in Lemma~\ref{lem:expo-ef1}.}
    \label{fig:flow1}
\end{figure}

    Note that the maximum flow in the above instance is at most $1$.
    Moreover, the lemma follows from the existence of a flow with value $1$ for every $M_j$, as we may construct the following randomized allocation.
    Assume we have a flow $f$ with value $1$, in which the flow on each edge $(u,v)$ is denoted by $f(u,v)$.
    Then, under each matching $\mu_{i'}$, the subjective good $M_j$ is allocated to agent $i$ with probability $\frac{f(i,\mu_{i'})}{\lambda_{i'}}$.
    Each integral allocation satisfies the constraint in Lemma~\ref{lem:expo-ef1} due to our edge construction between agents $\{1,\ldots,t\}$ and matchings $\{\mu_1,\ldots,\mu_{t'}\}$.
    The marginal probability that agent $i$ is allocated $M_j$ is given by $\sum_{i'=1}^{t'}(\lambda_{i'}\cdot\frac{f(i,\mu_{i'})}{\lambda_{i'}})=\sum_{i'=1}^{t'}f(i,\mu_{i'})$, which equals $f(s,i)$ by the flow conservation, which equals $\frac{1}{t}$ in any flow with value $1$.
    This matches the fractional allocation in which each agent $i\in T_j$ receives a fractional $\frac1t$ of $M_j$.

    We now prove the maximum flow is indeed $1$ by showing the instance has a minimum cut with value $1$.
    Each edge $(i,\mu_{i'})$ for $i\in[t]$ and $i'\in[t']$ will never appear in any minimum cut as it has capacity $\infty$.
    Therefore, any minimum cut contains only edges in $\{(s,i)\}_{i\in[t]}$, and edges in $\{(\mu_{i'},s')\}_{i'\in[t']}$.

    Consider any potential minimum cut.
    Define $U\subseteq \{1,\ldots,t\}$ such that the edge $(s,i)$ is included in the cut for each $i\in U$.
    If $U=T_j$, the cut has value at least $1$.
    Suppose $|U|=u<t$, then the cut involving edges in $\{(s,i)\}_{i\in[t]}$ has value $\frac{u}{t}$.
    For each edge $(\mu_{i'},s')$ where $i'\in [t']$, it should be included in the cut if $\mu_{i'}$ is an out-neighbor of any vertex in $T_j\setminus U$.
    This happens when all agents in $T_j$ are unmatched in $\mu_{i'}$, which corresponds to $K_{\mu_{i'}}\cap T_j=\emptyset$ (recall that $K_{\mu_{i'}}$ denotes the set of agents who are matched under $\mu_{i'}$), or there exists some agent in $T_j\setminus U$ that is matched in $\mu_{i'}$, which corresponds to $K_{\mu_{i'}}\cap T_j\nsubseteq U$.
    Since
    \[
    \mathbb{P}_{\mu\sim\lambda}\!\left(K_{\mu}\cap T_j\neq\emptyset \text{ and } K_{\mu}\cap T_j\subseteq J\right)\le \frac{|J|}{|T_j|}
    \]
    for each $J\subseteq T_j$ by the Hall-type condition, by taking $J$ to be $U$, we have
    \[
    \mathbb{P}_{\mu\sim\lambda}\!\left(K_{\mu}\cap T_j=\emptyset \text{ or } K_{\mu}\cap T_j\nsubseteq  U\right)\ge 1-\frac{|U|}{|T_j|} =1-\frac{u}{t}.
    \]
    It implies that the cut involving edges in $\{(\mu_{i'},s')\}_{i'\in[t']}$ has value at least $1-\frac{u}{t}$.
    Combining together, the minimum cut has value $1$.
\end{proof}

In the following, it suffices to show the lottery over the first-round allocations satisfying the Hall-type condition (\ref{eqn:hall}) is achievable.

\subsection{Decompositions Satisfying the Hall-Type Condition (\ref{eqn:hall})}\label{sect:hall}

In this section, we prove that decompositions satisfying the Hall-type condition (\ref{eqn:hall}) always exist.  We begin with the case of a single interest set \(T\), because it already contains the main ideas of the argument. The proof has three ingredients:
\begin{itemize}
    \item reformulate the Hall-type condition as a weighted inequality,
    \item use a minimax argument to reduce the problem to a fixed weight vector,
    \item construct the required distribution via a network flow based argument.
\end{itemize}
We then state the extension to multiple pairwise disjoint interest sets. The conceptual ideas are the same in the multi-set case, but carrying out the flow construction simultaneously for several sets introduces substantially more bookkeeping and notation. Since the extension is technically heavier rather than conceptually different, we defer its full proof to the appendix.

A key ingredient needed for our decomposition result is the Sion's minimax theorem, stated below.
\begin{theorem}[\cite{Sion58}]\label{thm:Sion}
Let \(P\subset \RR^p\)  and \(Q\subset \RR^q\) be compact convex sets. If $g:P \times Q\rightarrow \RR$ is a function such that:
\begin{itemize}
    \item for each fixed $x_0\in P$, the function $y\mapsto g(x_0,y)$  is continuous and concave;
    \item for each fixed $y_0\in Q$, the function $x\mapsto g(x,y_0) $ is continuous and convex.
\end{itemize}
Then,
\[
\min_{x\in P}\max_{y\in Q} g(x,y)
=
\max_{y\in Q}\min_{x\in P} g(x,y).
\]
\end{theorem}

Another useful ingredient is the following standard flow-integrality lemma, which lets us pass from a feasible fractional flow to a probability distribution over integral ones.
\begin{lemma}\label{lem:flow}
Let \( \mathcal{N}=(V,E)\) be a flow network.
For each arc \(a\in E\), let \(\ell_a,u_a\in \mathbb{Z}\) satisfy \(\ell_a\le u_a\), and let \(r\in \mathbb{Z}\).

Define the set of feasible \(s\)-\(t\) flows of value \(r\) by
\[
\mathcal{F}(\mathcal{N},\ell,u,r)
=
\left\{
f\in \mathbb{R}^E :
\begin{array}{l}
\ell_a\le f_a\le u_a \quad \forall a\in E,\\[2pt]
\displaystyle
\sum_{a\in\delta^-(v)} f_a-\sum_{a\in\delta^+(v)} f_a=
\begin{cases}
-r,& v=s,\\
r,& v=t,\\
0,& v\in V\setminus\{s,t\}
\end{cases}
\end{array}
\right\}.
\]
If \(\mathcal{F}(\mathcal{N},\ell,u,r)\neq\varnothing\), then every \(f\in \mathcal{F}(\mathcal{N},\ell,u,r)\) can be written as a convex combination of integral feasible \(s\)-\(t\) flows of value \(r\).
\end{lemma}

\subsubsection*{Single Interest Set $T$}

\begin{theorem}\label{thm:fixed-T}
Let \(X\in[0,1]^{n\times m}\) be an arbitrary matrix, and let \(r\in\mathbb{Z}\) satisfy \(0\le r<n\).  Suppose that each row of \(X\) sums to exactly \(r/n\), and each column sums to at most $1$. Then, for any fixed nonempty set \(T\subseteq [n]\), there exists a decomposition
\[
X=\sum_{\mu}\lambda_{\mu}\,\mathbf{1}_{\mu}
\]
into size-\(r\) matchings such that, for every subset \(J\subseteq T\),
\[
\mathbb{P}_{\mu\sim\lambda}\!\left(K_{\mu}\cap T\neq\emptyset \text{ and } K_{\mu}\cap T\subseteq J\right)\le \frac{|J|}{|T|}.
\]
Here, \(K_\mu\) denotes the random variable corresponding to the set of rows in \([n]\) that are matched under \(\mu\).
\end{theorem}

At a high level, \Cref{thm:fixed-T} says that we can decompose \(X\) so that the matched agents from \(T\) do not concentrate too heavily on any particular subset \(J\subseteq T\). Rather than working directly with these events, we first prove a weighted version that is more flexible and better suited to a minimax argument.

For the rest of this subsection, \(X\) and \(r\) are fixed as in \Cref{thm:fixed-T}.

\begin{definition}
Let $w \in \mathbb{R}_+^T$ be a nonnegative weight vector. Define the set function $\varphi_w : 2^T \to \mathbb{R}_+$ by
\[
\varphi_w(S) :=
\begin{cases}
\min_{i \in S} w_i, & \text{if } S \neq \emptyset, \\
0, & \text{otherwise}.
\end{cases}
\qquad \text{for all } S \subseteq T.
\]
\end{definition}

Let \(\Pi(X)\) denote the polytope of all probability distributions over size-\(r\) matchings that implement \(X\). The next proposition is the main technical step; \Cref{thm:fixed-T} will follow from it by choosing an appropriate weight vector.

\begin{proposition}\label{prop:key-T}
There exists a probability distribution $\lambda^*$ over matchings of size $r$ such that $\lambda^* \in \Pi(X)$ and
\[
\mathbb{E}_{\mu\sim\lambda^*}\!\left[\varphi_w(K_\mu \cap T)\right]
\le \frac{1}{|T|}\sum_{i\in T} w_i
\qquad \text{for every } w \in \mathbb{R}_+^T.
\]
\end{proposition}
\begin{proof}

For $\lambda\in \Pi(X)$ and $w\in \RR^T_+$, define
\[
g(\lambda,w):=\mathbb{E}_{\mu\sim\lambda}\!\left[\varphi_w(K_\mu \cap T)\right]
- \frac{1}{|T|}\sum_{i\in T} w_i.
\]
The statement of the proposition is equivalent to showing that $\min_{\lambda\in \Pi(X)} \max_{w\in \RR^T_+}g(\lambda,w) \leq 0$.

Observe that, for any constant $c>0$, we have $g(\lambda,cw) = c g(\lambda,w)$, i.e., the function $g(\lambda,w)$ is positively homogeneous in $w$. Hence,  it suffices to consider weight vectors that are in the standard simplex
\[
\Delta(T):=\Bigl\{w\in \RR_+^T : \sum_{i\in T} w_i \leq 1\Bigr\}.
\]
This domain normalization is useful since $\Delta(T)$ is compact.

The function $g(\lambda,w): \Pi(X) \times \Delta(T) \rightarrow \RR$ satisfies:
\begin{itemize}
    \item For each fixed $\lambda_0 \in \Pi(X)$, the map $w \mapsto g(\lambda_0,w)$ is continuous and concave. The concavity follows because, for every fixed subset $S \subseteq T$, the map $w \mapsto \varphi_w(S)$ is concave, and taking expectations preserves concavity. Moreover, subtracting a linear function from a concave function preserves concavity.
    \item  For each fixed $w_0 \in \Delta(T)$, the map $\lambda \mapsto g(\lambda,w_0)$ is linear. This is because $\varphi_{w_0}(K_\mu\cap T)$ is a constant number for each $K_\mu$.
\end{itemize}
The above properties imply that \(g(\lambda,w)\) satisfies the assumptions of \Cref{thm:Sion}. Hence,
\begin{align}\label{eq:minimax}
    \min_{\lambda\in \Pi(X)} \max_{w\in \Delta(T)} g(\lambda,w)
=
\max_{w\in \Delta(T)} \min_{\lambda\in \Pi(X)} g(\lambda,w).
\end{align}
The main advantage of this minimax identity is that, instead of having to construct a single distribution \(\lambda' \in \Pi(X)\) satisfying \(g(\lambda',w)\le 0\) for all $w$ simultaneously, it suffices to prove that for every fixed \(w\in \Delta(T)\) there exists some \(\lambda_w\in \Pi(X)\) such that \(g(\lambda_w,w)\le 0\). This is a much more manageable task. Indeed,  \Cref{lem:Single-T} shows that such a distribution can be constructed for every \(w\in \Delta(T)\).

By \Cref{lem:Single-T}, we have
\(
\max_{w\in \Delta(T)} \min_{\lambda\in \Pi(X)} g(\lambda,w)\le 0
\).
Combining with \Cref{eq:minimax}, we obtain
\(
\min_{\lambda\in \Pi(X)} \max_{w\in \Delta(T)} g(\lambda,w)\le 0
\).
By the positive homogeneity of \(g(\lambda, w) \) in \(w\), this is equivalent to
\[
\min_{\lambda\in \Pi(X)} \max_{w\in \mathbb{R}_+^T} g(\lambda,w)\le 0,
\]
which proves the proposition.
\end{proof}

\begin{lemma}\label{lem:Single-T}
For each fixed \(w\in \Delta(T)\), there exists \(\lambda_w\in \Pi(X)\) such that \(g(\lambda_w,w)\le 0\).
\end{lemma}
\begin{proof}
Fix \(w \in \Delta(T)\).
Note that if \(\frac{r|T|}{n} \leq 1\), then any $\lambda\in \Pi(X)$ satisfies the desired bound. To see this,  note that
\[
\varphi_w(S)\le \sum_{i\in S} w_i,
\]
and hence,
\[
\EE_{\mu\sim\lambda}[\varphi_w(K_\mu \cap T)]
\le
\EE_{\mu\sim\lambda}\left[\sum_{i\in T} w_i \ \mathbb{I}[i\in K_\mu]\right]
=
\sum_{i\in T} w_i\,\PP_{\mu\sim\lambda}(i\in K_{\mu})
=
\frac{r}{n} \sum_{i\in T} w_i
 \le \frac{1}{|T|} \sum_{i\in T} w_i
\]
The second equality holds because each row of \(X\) sums to exactly \(\frac{r}{n}\), and \(\lambda\) implements \(X\); therefore, for every \(i\), the probability that agent \(i\) is matched is exactly \(\frac{r}{n}\). Hence, whenever  \(\frac{r|T|}{n} \leq 1\), then any $\lambda\in \Pi(X)$ satisfies $g(\lambda,w) \leq 0$.

For the remainder of the proof, assume that \(\frac{r|T|}{n} > 1\). Relabel the indices of \(T\) so that their weights are ordered as
\[
w_1 \ge w_2 \ge \cdots \ge w_{|T|}.
\]
The distribution $\lambda_w$ that we construct will depend solely on the ordering of $w$.

Next, for each $k\in \{1,\dots, |T|\}$, define \( \delta_k := w_k - w_{k+1} \) with the convention that \(w_{|T|+1}=0\), and set \(J_k := \{1,\dots,k\} \).
With this notation, every nonempty subset \(S \subseteq T\) satisfies
\begin{align}\label{eq:phi-Identity}
    \varphi_w(S)=\sum_{k=1}^{|T|} \delta_k\, \mathbb{I}[S\subseteq J_k].
\end{align}
Indeed, if \(q\) is the largest index in \(S\), then \(\varphi_w(S)=w_q\), since the weights are arranged in nonincreasing order. Also, \(S \subseteq J_k\) holds exactly when \(k \ge q\). Therefore,
\[
\sum_{k=1}^{|T|} \delta_k\,\mathbb{I}[S\subseteq J_k]
=
\sum_{k=q}^{|T|} \delta_k
=
\sum_{k=q}^{|T|}(w_k-w_{k+1})
=
w_q
=
\varphi_w(S),
\]
as required.

Substituting \Cref{eq:phi-Identity} into the expectation yields, for any \(\lambda\in \Pi(X)\),
\begin{align}\label{id:exp-eq-T}
    \EE_{\mu\sim\lambda}[\varphi_w(K_\mu \cap T)] &=   \sum_{k=1}^{|T|} \delta_k \ \PP_{\mu\sim\lambda} \!\left(K_{\mu}\cap T\neq\emptyset \text{ and } K_{\mu}\cap T\subseteq J_k\right).
\end{align}
Therefore, it is enough to find \(\lambda_w \in \Pi(X)\) such that
\begin{align}\label{ineq:T}
    \PP_{\mu\sim\lambda_w} \!\left(K_{\mu}\cap T\neq\emptyset \text{ and } K_{\mu}\cap T\subseteq J_k\right) \leq \frac{k}{|T|} \qquad \text{ for each } k\in \{1,...,|T|\}.
\end{align}
Indeed, once \Cref{ineq:T} is established, the lemma follows by substituting the bound to \Cref{id:exp-eq-T} as shown below,
 \begin{align*}
    \EE_{\mu\sim\lambda_w}[\varphi_w(K_\mu \cap T)] &\leq \sum_{k=1}^{|T|} \delta_k \frac{k}{|T|}
    = \frac{1}{|T|} \sum_{k=1}^{|T|} k \ \delta_k
    = \frac{1}{|T|} \sum_{k=1}^{|T|} w_k.
\end{align*}
Thus \(g(\lambda_w,w)\le 0\), as needed.

We have therefore reduced the lemma to constructing a distribution $\lambda_w$ that satisfies the  chain of bounds in \Cref{ineq:T}. We obtain such decomposition $\lambda_w$ of $X$ by a carefully constructed flow network with a feasible flow and subsequently applying \Cref{lem:flow}.

\begin{figure}[t]
    \centering
\begin{tikzpicture}[
    >=Stealth,
    every node/.style={font=\sffamily},
    nd/.style={
        circle, draw=black!70, fill=white,
        line width=0.8pt, minimum size=24pt,
        font=\sffamily\small
    },
    arr/.style={->, line width=0.7pt, black!70}, scale = 0.7
]

\def\srcX{-2}
\def\chainX{3}
\def\rowX{6}
\def\itemX{10}
\def\sinkX{13}

\def\vgap{1.8}

\node[nd] (s) at (\srcX, -4.5) {$s$};

\node[nd] (t1) at (\chainX, 0) {$\tau_1$};
\node[nd] (t2) at (\chainX, -\vgap) {$\tau_2$};
\node[nd] (t3) at (\chainX, -2*\vgap) {$\tau_3$};

\node[font=\sffamily\bfseries, text=black!60] at (\chainX, -3*\vgap + 0.3) {$\vdots$};
\node[nd] (tT) at (\chainX, -3*\vgap - 0.8) {$\tau_{|T|}$};

\draw[arr] (t1) -- (t2);
\draw[arr] (t2) -- (t3);
\draw[arr] (t3) -- ++(0, -1.3);

\draw[arr] (s) -- (t1);

\node[nd] (r1) at (\rowX, 0) {$1$};
\node[nd] (r2) at (\rowX, -\vgap) {$2$};
\node[nd] (r3) at (\rowX, -2*\vgap) {$3$};

\node[font=\sffamily\bfseries, text=black!60] (rdots) at (\rowX, -3*\vgap + 0.7) {$\vdots$};
\node[nd] (rT) at (\rowX, -3*\vgap - 0.8) {$|T|$};

\begin{scope}[on background layer]
    \node[rounded corners=6pt, draw=black!40, line width=0.7pt,
          inner xsep=10pt, inner ysep=10pt,
          fit=(r1)(r2)(r3)(rdots)(rT)] {};
\end{scope}

\draw[arr] (t1) -- (r1);
\draw[arr] (t2) -- (r2);
\draw[arr] (t3) -- (r3);
\draw[arr] (tT) -- (rT);

\node[font=\sffamily\bfseries, text=black!60] at (\rowX, -3*\vgap - 2.3) {$\vdots$};
\node[nd] (rn) at (\rowX, -3*\vgap - 3.5) {$n$};

\draw[arr] (s) -- (rn);

\node[nd] (c1) at (\itemX, 0.8) {$c_1$};
\node[nd] (c2) at (\itemX, -1.2) {$c_2$};

\node[font=\sffamily\bfseries, text=black!60] at (\itemX, -3.2) {$\vdots$};
\node[nd] (cmm2) at (\itemX, -5.2) {$c_{m\!-\!2}$};
\node[nd] (cmm1) at (\itemX, -7.0) {$c_{m\!-\!1}$};
\node[nd] (cm) at (\itemX, -3*\vgap - 3.5) {$c_m$};

\draw[arr] (r1) -- (c1);
\draw[arr] (r1) -- (c2);
\draw[arr] (r2) -- (c1);
\draw[arr] (r2) -- (c2);
\draw[arr] (r3) -- (c2);
\draw[arr] (rT) -- (c2);
\draw[arr] (rT) -- (cmm1);
\draw[arr] (rn) -- (cmm2);
\draw[arr] (rn) -- (cm);

\node[nd] (t) at (\sinkX, -2.5) {$t$};

\draw[arr] (c1) -- (t);
\draw[arr] (c2) -- (t);
\draw[arr] (cmm2) -- (t);
\draw[arr] (cmm1) -- (t);
\draw[arr] (cm) -- (t);

\end{tikzpicture}

    \caption{The flow network used in the proof of \Cref{lem:Single-T}. The interest set $T$ is highlighted by a rectangle}
    \label{fig:flowT}
\end{figure}
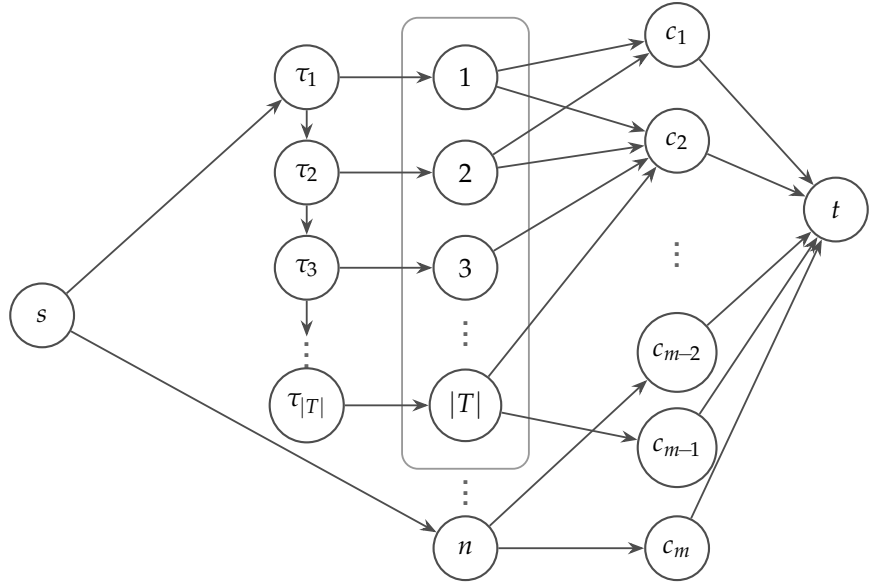

\medskip
\noindent\textit{Biased Flow Network.}
We construct a directed flow network as shown in \Cref{fig:flowT}. Formally, we
build a directed network with source $s$, sink $t$, chain nodes
$\tau_1,\dots,\tau_{|T|}$, row nodes $i$ for $i\in[n]$, and item nodes $c_j$ for
$j\in[m]$.
The arcs are given as follows:
\begin{itemize}[leftmargin=2em]
\item $s\to \tau_1$ with lower and upper capacities $\ell = \lfloor \frac{r}{n} |T| \rfloor $ and $u= \lceil \frac{r}{n} |T| \rceil$;
\item for each $k\in \{1,\dots,|T|-1\}$, an arc $\tau_k\to \tau_{k+1}$ with lower and upper capacities $\ell = \lfloor \frac{r}{n} (|T|-k) \rfloor $ and $u= \lceil \frac{r}{n} (|T|-k) \rceil$ ;
\item for each $k\in \{1,\dots,|T|\}$, an arc $\tau_k\to k$ with lower and upper capacities 0 and 1;
\item for each row $i\notin T$, an arc $s\to i$ with lower and upper capacities 0 and 1;
\item for each pair $(i,j)$ with $X_{ij}>0$, an arc $i\to c_j$ with lower and upper capacities 0 and 1;
\item for each item node $c_j$, an arc $c_j\to t$ with lower and upper capacities 0 and 1.
\end{itemize}
Note that lower and upper bounds on the flows on each arc are integral.
We now define a feasible fractional $s$-$t$ flow $f$ on this network with value $r$:
\[
f(s,\tau_1)=\frac{r}{n} |T|,
\qquad
f(\tau_k,\tau_{k+1})=\frac{r}{n}(|T|-k),
\qquad
f(\tau_k,k)=\frac{r}{n}
\quad \text{ for each }k\in \{ 1,...,|T|\},
\]
\[
f(s,i)=\frac{r}{n}
\qquad\text{for each } i\notin T,
\]
\[
f(i,c_j)=X_{ij},
\qquad
f(c_j,t)=\sum_{i=1}^n X_{ij}.
\]
Note that the flow $f$ is feasible since $0\leq \frac{r}{n}\leq 1$, and since each column sum of $X$ is at most 1, we have $0\leq \sum_{i=1}^n X_{ij}\leq 1$. Flow conservation holds since each row of $X$ sums to exactly $\frac{r}{n}$.
The total flow value is $\sum_{i\in [n]}\sum_{j\in [m]} X_{ij}=r$. Thus, the network and the flow satisfy the conditions of \Cref{lem:flow}, hence can be expressed as a distribution $\lambda_w$ over integral feasible flows of value $r$.

Now fix any integral feasible flow in this decomposition. This flow defines a matching of size $r$, since every row node has inflow at most $1$ and therefore sends flow along at most one row--item arc, while every item node has capacity $1$ on its outgoing arc to the sink and hence is incident to at most one matched row. Hence, the distribution $\lambda_w$ over the integral feasible flows of value $r$ also gives distribution over size $r$ matchings.

Moreover, this distribution over matchings implements \(X\) i.e., \(\lambda_w\in \Pi(X)\). Indeed, for every pair \((i,j)\) with \(X_{ij}>0\), the edge \((i,j)\) is included in the matching with probability \(f(i,c_j)=X_{ij}\), since the integral flows decompose the fractional flow.

Let $R_{k}:=\{k,...,|T|\}$ be the suffix of $T$ starting at index $k$. Observe that for any $\mu$ in the support of $\lambda_w$, the set of rows $K_\mu$ that are matched under $\mu$ satisfies
\begin{align}\label{eq:suffixT}
    |K_\mu \cap R_{k+1}| \in \left\{   \left\lfloor \frac{r}{n} (|T|-k) \right\rfloor , \left\lceil \frac{r}{n} (|T|-k) \right\rceil \right\} \quad \text{ for each } k\in \{0,...,|T|-1\}.
\end{align}
To see this, first consider $k\geq 1$, note that the arc $\tau_k\to \tau_{k+1}$ has lower and upper capacities $\ell = \lfloor \frac{r}{n} (|T|-k) \rfloor $ and $u= \lceil \frac{r}{n} (|T|-k) \rceil$. Since flow corresponding to $\mu$ is integral, the amount of flow on this arc must be an integer between these two bounds. Moreover, by flow conservation, this flow is exactly $|K_\mu \cap R_{k+1}|$ because each matched row contributes one unit of flow and each row can be matched to at most one item node.
When $k=0$, the arc $s\rightarrow \tau_1$ has  lower and upper capacities $ \lfloor \frac{r}{n} |T| \rfloor $ and $ \lceil \frac{r}{n} |T| \rceil$, by the same argument, the flow on this arc is exactly $|K_\mu \cap R_{1}|=|K_\mu \cap T|$. This establishes the identity \eqref{eq:suffixT}.

Recall that \(\frac{r|T|}{n} > 1\), and hence by  \eqref{eq:suffixT}, we see that $|K_\mu \cap T|\geq\lfloor \frac{r}{n} |T| \rfloor  \geq 1$. Thus, every matching $\mu$ in the support of $\lambda_w$ satisfies $K_{\mu}\cap T \neq \emptyset$. It follows that, for each $k\in \{1,...,|T|-1\}$, the following events are equivalent:
\[K_{\mu}\cap T\neq\emptyset \text{ and } K_{\mu}\cap T\subseteq J_k \iff |K_{\mu} \cap R_{k+1}|=0. \]

\noindent\textit{Verifying $\lambda_w$ satisfies \Cref{ineq:T}.} \Cref{ineq:T} is restated below.
\begin{align*}
    \PP_{\mu\sim\lambda_w} \!\left(K_{\mu}\cap T\neq\emptyset \text{ and } K_{\mu}\cap T\subseteq J_k\right) \leq \frac{k}{|T|} \qquad \text{ for each } k\in \{1,...,|T|\} \tag{\ref{ineq:T}}
\end{align*}
Note that \Cref{ineq:T} is trivial for \(k=|T|\), so we may assume that \(k\in\{1,\dots,|T|-1\}\). Define
\[
Y_{k+1}:=|K_\mu\cap R_{k+1}|.
\]
Note that this random variable takes values in $\left\{   \left\lfloor \frac{r}{n} (|T|-k) \right\rfloor , \left\lceil \frac{r}{n} (|T|-k) \right\rceil \right\}$ by identity \eqref{eq:suffixT}.
We now establish  \Cref{ineq:T} by considering the following two cases.

\medskip

\noindent\textit{Case 1: $\frac{r}{n} (|T|-k)\geq 1$.} In this case, we have that $Y_{k+1} = |K_\mu\cap R_{k+1}| \geq \lfloor \frac{r}{n} (|T|-k) \rfloor\geq 1$ for any matching $\mu$ in the support of $\lambda_w$. Hence,
\[
\PP_{\mu\sim\lambda_w} \!\left(K_{\mu}\cap T\neq\emptyset \text{ and } K_{\mu}\cap T\subseteq J_k\right) = \PP_{\mu\sim\lambda_w} \!\left( Y_{k+1}=0\right)=0.
\]

\medskip

\noindent\textit{Case 2:  $\frac{r}{n} (|T|-k)< 1$.}
Since \(\lambda_w\) implements \(X\), each row \(i\) is matched with probability exactly \(\frac{r}{n}\). As \(|R_{k+1}|=|T|-k\), it follows by linearity of expectation that
\(
\EE_{\mu\sim\lambda_w}[Y_{k+1}]
=
\EE_{\mu\sim\lambda_w}[|K_\mu\cap R_{k+1}|]
=
\frac{r}{n}(|T|-k).
\)
Furthermore, by identity \eqref{eq:suffixT} and the case distinction, we know that the random variable $Y_{k+1}$ takes values in $\{0,1\}$. Hence, we have $\PP_{\mu\sim\lambda_w} \!\left( Y_{k+1} =1 \right) = \EE_{\mu\sim\lambda_w}[|K_\mu\cap R_{k+1}|] =  \frac{r}{n}(|T|-k)$. We obtain,
\begin{align*}
    \PP_{\mu\sim\lambda_w} \!\left(K_{\mu}\cap T\neq\emptyset \text{ and } K_{\mu}\cap T\subseteq J_k\right) = \PP_{\mu\sim\lambda_w} \!\left( Y_{k+1}=0\right)
    = 1 - \frac{r}{n}(|T|-k)
    < 1- \frac{|T|-k}{|T|}
    =\frac{k}{|T|}.
\end{align*}
Here, the inequality follows from the assumption that \(\frac{r|T|}{n} > 1\).

Thus, in both of the cases \Cref{ineq:T} holds for the decomposition $\lambda_w$. This concludes the proof of the Lemma.
\end{proof}

\begin{proof}[Proof of \Cref{thm:fixed-T}]
We now return from the weighted statement to the original Hall-type condition.
We  show that $\lambda^*$ given in \Cref{prop:key-T} gives the desired distribution. By \Cref{prop:key-T}, we have $\lambda^* \in \Pi(X)$, and hence $X=\sum_{\mu}\lambda^*_{\mu}\,\mathbf{1}_{\mu}$.

Fix any subset \(J \subseteq T\), and let \(\widetilde{w} = \mathbf{1}_J\). That is, \(\widetilde{w}\) is the weight vector defined coordinatewise by
\[
\widetilde{w}_i =
\begin{cases}
1, & \text{if } i \in J,\\
0, & \text{if } i \notin J.
\end{cases}
\]
For this choice of weight vector, we have, for every $S \subseteq T$,
\[
\varphi_{\widetilde{w}}(S)=
\begin{cases}
1, & \text{if } S \neq \emptyset \text{ and } S \subseteq J,\\
0, & \text{otherwise}.
\end{cases}
\]
Therefore,
\begin{align*}
 \mathbb{P}_{\mu\sim\lambda^*}\!\left(K_{\mu}\cap T\neq\emptyset \text{ and } K_{\mu}\cap T\subseteq J\right) &=    \mathbb{E}_{\mu\sim\lambda^*}\!\left[\varphi_{\widetilde{w}}(K_\mu \cap T)\right] \\
&\le \frac{1}{|T|}\sum_{i\in T} \widetilde{w}_i \\
&= \frac{|J|}{|T|}.
\end{align*}
The inequality follows from \Cref{prop:key-T} applied to the weight vector $\widetilde{w}$. Since $J \subseteq T$ was arbitrary, this establishes the desired inequality.
\end{proof}

\subsubsection*{Multiple Disjoint Interest Sets}
We now extend the result to multiple disjoint interest sets, thereby establishing the Hall-type condition~(\ref{eqn:hall}). The proof is conceptually similar to the case of a single interest set, though it entails significant additional technical overhead in constructing the flow and applying the minimax argument. We therefore defer the proof to the appendix.

\begin{theorem}\label{thm:multi-T}
Let \(X\in[0,1]^{n\times m}\) be an arbitrary matrix, and let \(r\in\mathbb{Z}\) satisfy \(0<r<n\). Suppose that each row of \(X\) sums to exactly \(r/n\), and each column sums to at most one. Then, for any fixed family of pairwise disjoint sets \(\{T_1,\dots,T_{\ell}\}\) with \(T_i\subseteq [n]\), there exists a decomposition
\[
X=\sum_{\mu}\lambda_{\mu}\,\mathbf{1}_{\mu}
\]
into size-\(r\) matchings such that, for each \(i\in[\ell]\) and every  subset \(J\subseteq T_i\),
\[
\mathbb{P}_{\mu\sim\lambda}\!\left(K_{\mu}\cap T_i\neq\emptyset \text{ and } K_{\mu}\cap T_i\subseteq J\right)\le \frac{|J|}{|T_i|}.
\]
Here, \(K_\mu\) denotes the random variable corresponding to the set of rows in \([n]\) that are matched under \(\mu\).
\end{theorem}

\section{Easy Case: At Most \texorpdfstring{$n$}{n} Objective Chores}
\label{sect:easycase}

The case with at most $n$ objective chores after bundling is handled via a different approach.
Consider the instance with a set of subjective goods $G$ and a set of objective chores $Z$.
For each subjective good $g\in G$ after bundling, it is referred to as a \emph{meta-subjective good} if it is packed, that is, it contains multiple subjective goods from the original instance.
We further apply the following update rules iteratively, until neither is applicable.
Note that some meta-subjective goods in $G$ may get unbundled.
\begin{itemize}
    \item[(i)] Let $M_j=\{g_1,\ldots,g_K\}$ be a meta-subjective good in $G$ for $K\ge 2$, where each $g_k$ is a subjective good from the original instance.
    While there exists an agent $i\in N$ with $v_i(g) \ge 0$ such that for every $g\in M_j$ with $v_i(g)\ge 0$, we have $v_i(M_j\setminus\{g\})\ge 0$, arbitrarily remove $g$ with $v_i(g)\ge 0$ from $M_j$ and add $g$ to $G$.
    \item[(ii)] While $Z\neq\emptyset$, if there exists an agent $i\in N$ and an objective chore $c\in Z$ such that $v_i(G_i\cup \{c\})\ge 0$ for $G_i=\{g\in G:v_i(g)\ge 0\}$, we find a minimal subset $S$ of $G_i$ such that $v_j(S\cup\{c\})\ge 0$ for some agent $j$ (which may or may not be $i$).
    We consider $S\cup\{c\}$ as a single packed subjective good $g'$, that is, we remove $S$ and $c$ respectively from $G$ and $Z$, and add the $g'$ to $G$.
\end{itemize}

Note that the above process always terminates as objective chores are only consumed throughout the update.
Additionally, the number of objective chores remains at most $n$.
We may easily verify that the resulting instance satisfying the following proposition.
\begin{proposition}\label{prop:le_n}
    The above process outputs a set of subjective goods $G$ and a set of objective chores $Z$ that have the following properties:
    \begin{enumerate}
        \item each meta-subjective good $M_j\in G$ is \emph{good-minimal}, that is, for each agent $i\in N$, there exists $g\in M_j$ such that $v_i(M_j\setminus\{g\})<0$, where $g$ is a subjective good from the original instance;
        \item when $Z\neq\emptyset$, the instance is a \emph{small-goods-only} instance, that is, for each agent $i\in N$ and each objective chore $c\in Z$, it holds that $$-v_i(c)> \sum_{g\in G:v_i(g)\ge 0}v_i(g).$$
        Note that chore-maximality of each subjective good $g\in G$ is implied from above, that is, $v_i(g\cup\{c\})<0$ for each $i\in N$ and each $c\in Z$.
    \end{enumerate}
\end{proposition}
\begin{proof}
    If 1 is not satisfied, we should continue to apply rule (i); if 2 is not satisfied, we should continue to apply rule (ii).
\end{proof}

We define a randomized allocation as follows.
Let $|Z|=t$ and we label the objective chores in $Z$ in an arbitrary order $c_1,\ldots,c_t$.
Firstly, uniformly at random sample a set $T$ of $t$ agents from $\binom{n}{t}$ possible outcomes ($T=\emptyset$ if $Z=\emptyset$), and uniformly at random allocate the $t$ objective chores in $Z$ to the agents in $T$ such that each agent receives exactly one objective chore from $t!$ possible outcomes.
We assume without loss of generality that the agent receiving item $c_i$ has index $i$.
In the second step, we implement a \emph{serial dictatorship algorithm} for agents in $T$ in the order $1,\ldots,t$: let each agent $i\in T$ further receive all the remaining subjective goods which $i$ values non-negatively.
In other words, let $G_i=\{g\in G:v_i(g)\geq0\}$, each agent $i$ is further allocated $G_i\setminus\bigcup_{i'=1}^{i-1}G_{i'}$.
In the last step, a PS-lottery is implemented among the agents $N\setminus T$ and the remaining subjective goods, with the dummy items added to ensure no one receives an item with a negative value.

The following lemma immediately implies Theorem~\ref{thm:main} for the case of $|Z|\le n$.
\begin{lemma}
    The above randomized allocation is \exan EF and \expo EF$1$.
\end{lemma}
\begin{proof}
We first show that each of the above integral allocations $(A_1,\ldots,A_n)$ satisfies \expo EF1, which follows straightforwardly from the following observations and the property of PS-lottery.
\begin{enumerate}
    \item For each $i\in T$ who receives an objective chore $c_i$, we have $v_i(A_i\setminus\{c_i\})\geq0$ since $c_i$ is the only negatively-valued item.
    Note that each objective chore in the updated instance is also an objective chore in the original instance.
    \item For each $i\in N\setminus T$, we have $v_i(A_i)\geq 0$ since $i$ only receives subjective goods with non-negative value during the PS-lottery.
    \item For each $i\in N$ and $j\in T$, we have $v_i(A_j)<0$, due to the small-goods-only property in Proposition~\ref{prop:le_n}.
    \item For each $i\in T$ and $j\in N\setminus T$, we have $v_i(A_j)<0$ as all non-negatively-valued subjective goods for agent $i$ have been allocated in the second step (the serial dictatorship algorithm).
\end{enumerate}
For two agents in $T$, the EF1 property is guaranteed by 1 and 3.
For two agents $i$ and $j$ in $N\setminus T$, guaranteed by the PS-lottery, we have $v_i(A_i)\ge v_i(A_j\setminus\{g\})$ for some $g\in A_j$, where $g$ is a subjective good in the updated instance (after iteratively applying (i) and (ii) described at the beginning of this section).
If $g$ is also a subjective good in the original instance, EF$1$ trivially holds.
Otherwise, if $g$ is a meta-subjective good, there exists a subjective good $g'\in g$ from the original instance such that $$v_i(A_j\setminus\{g'\})= v_i(A_j\setminus \{g\})+v_i(g\setminus\{g'\})<v_i(A_j\setminus \{g\})\le v_i(A_i),$$
where the first inequality holds due to good-minimality of each meta-subjective good.
An agent in $T$ will not EF1-envy an agent in $N\setminus T$ because of 1 and 4.
An agent in $N\setminus T$ will not envy an agent in $T$ due to 2 and 3.

It then remains to show \exan EF.
Notice that \exan EF among agents in $N\setminus T$ is known for the PS-lottery in the third step.
Since each item allocated to $N\setminus T$ in the third step has a negative value to each agent in $T$ (due to the serial dictatorship nature of the second step), EF is guaranteed for all agents in the third step.
It then suffices to show that \exan EF is satisfied after the first two steps of the algorithm.
We consider two arbitrary agents $i$ and $j$ and show that $i$ will not envy $j$ in the \exan sense after the first two steps.
In every outcome where both $i$ and $j$ do not belong to $T$, \exan EF is satisfied as both of them receive no item at all.
For the remaining possible outcomes, we use a coupling argument that ``pairs'' the outcomes (notice that the probabilities for all the outcomes are equal) by groups of two, and prove that agent $i$ will not envy agent $j$ in expectation for every pair of outcomes.

We first consider the outcome where both $i$ and $j$ belong to $T$, and agent $i$ and $j$ receive objective chores $c_{i}$ and $c_{j}$ respectively in the first step.
This outcome is paired with the outcome where the ranking of $i$ and $j$ is swapped while the remaining part of the allocation in the first step is the same, i.e., the outcome where $i$ receives $c_{j}$ and $j$ receives $c_{i}$ in the first step.
We assume $i<j$ without loss of generality when considering these two outcomes.
For the objective chore received, the average utility for agent $i$ is $\frac12(v_i(c_{i})+v_i(c_{j}))$, and this is also how agent $i$ views agent $j$'s allocation in the first step.
Thus, \exan EF holds for now.
Next, let $H_i$ be the set of all remaining subjective goods that agent $i$ values non-negatively after agents who have received $c_1,\ldots,c_{i-1}$ take their non-negatively-valued subjective goods in the second step.
By our algorithm, in the outcome where agent $i$ receives $c_{i}$, agent $i$ receives all items in $H_i$ and all items received by agent $j$ in the second step are chores to agent $i$; in the paired outcome where agent $j$ receives $c_{i}$, agent $j$ receives only a subset of $H_i$ in agent $j$'s perspective and the remaining items received by agent $j$ have negative values in agent $i$'s perspective.
The average utility for agent $i$ is no less than $\frac12 v_i(H_i)$, and the average utility that agent $i$ views agent $j$ is no more than $\frac12v_i(H_i)$.
Therefore, \exan EF also holds for the second step.

We finally consider the case where one of $i,j$ belongs to $T$.
The outcome where agent $i$ receives $c_{i}$ and agent $j$ receives nothing in the first round is paired with the outcome where agent $i$ receives nothing and agent $j$ receives $c_{i}$, where the allocation for the remaining agents in the first step is fixed.
The analysis is similar to before.
For the first step, the average utility for agent $i$ is $\frac12 v_i(c_{i})$, and agent $i$ views agent $j$'s bundle in the same way.
For the second step, if agent $i$ receives $c_{i}$, agent $i$ will further receive a non-negatively-valued bundle $H_i$, and all items received by agent $j$ are chores to agent $i$; if agent $i$ receives nothing, agent $j$ receives only a subset of $H_i$ which are goods in agent $j$'s perspective and the remaining items received by agent $j$ are chores in agent $i$'s perspective.
This concludes \exan EF.
\end{proof}

\section{Discussion}

In this paper, we have studied the problem of fairly allocating a mix of goods and chores among agents who have additive utilities over the set of indivisible items.
We show that there always exists a randomized allocation (i.e., a lottery over deterministic allocation) that simultaneously satisfies \exan \EF{} and \expo \EFOne.

All components of our construction can be implemented in polynomial time except for computing the probabilistic Hall-type decomposition in Section~\ref{sect:hall}.  A direct implementation that explicitly enumerates the relevant matchings and constraints gives an exponential-time algorithm.  On the other hand, the decomposition's minimax formulation and its polynomial-time network-flow subproblem for each fixed weight vector suggest that an oracle-based or more compact implementation may yield a polynomial-time algorithm.  Establishing such an implementation is beyond the scope of this work, whose focus is the compatibility of the ex-ante and ex-post fairness guarantees.

An interesting direction is to explore further avenues for applying our techniques. Given the generality of the probabilistic Hall-type decomposition, we believe our framework has additional applications both within fair division and beyond. In particular, since the decomposition allows us to choose an arbitrary set and bound the concentration of matchings over any of its subsets, it has potential applications to a wide range of correlated matching problems.

Our work focuses exclusively on fairness.
In future research, an immediate next step could be understanding the interplay between fairness and economic efficiency and charting the boundary between impossibilities and positive results when applying the BoBW approach.
Another interesting direction is to consider agents' utility functions beyond additive utilities.
Nevertheless, a solid understanding for the goods- and chores-only settings would be the first step.
Going beyond additive utilities, it is not known whether \exan \EF{} and \expo \EFOne are compatible or not even for the goods-only or chores-only setting.

It would also be interesting to consider unequal entitlements of the agents and study the weighted setting of mixed goods and chores.
As discussed in Section~\ref{sec:related}, for goods- and chores-only settings, weighted \EF{} is compatible with a version of its relaxation.
Even considering only \expo fairness in the weighted setting, it is intriguing if weighted \EFOne (with proper adaptation for mixed goods and chores) or its relaxation can always be satisfied~\citep[Open problem~8]{Suksompong25}.

\section*{Acknowledgements}

This work was partially supported
by the NSF-CSIRO grant on ``Fair Sequential Collective Decision Making'',
by the ARC Laureate Project FL200100204 on ``Trustworthy AI'',
and by JST ERATO Grant Number JPMJER2301.

\bibliographystyle{plainnat}
\bibliography{bibliography}

\appendix
\section{Omitted Proofs}
\subsection*{Proof of \Cref{lem:flow}}

    Let \(B\in\mathbb{R}^{V\times E}\) be the node-arc incidence matrix of the directed network \(N\), defined by
\[
B_{v,a}=
\begin{cases}
1,& \text{if } a \text{ enters } v,\\
-1,& \text{if } a \text{ leaves } v,\\
0,& \text{otherwise}.
\end{cases}
\]
Let \(b\in\mathbb{Z}^V\) be the demand vector given by
\[
b_v=
\begin{cases}
-r,& v=s,\\
r,& v=t,\\
0,& v\in V\setminus\{s,t\}.
\end{cases}
\]
Then
\[
\mathcal{F}(\mathcal{N},\ell,u,r)=\{f\in\mathbb{R}^E : Bf=b,\ \ell\le f\le u\}.
\]

It is a standard fact that the node-arc incidence matrix \(B\) of a directed graph is totally unimodular. Since \(b,\ell,u\) are integral, the polyhedron $P=\{f\in\mathbb{R}^E : Bf=b,\ \ell\le f\le u\}$ is therefore integral.

\subsection{Proof of \texorpdfstring{\Cref{thm:multi-T}}{the multiple-set theorem}}\label{append:proof-thm-multi-T}

The proof closely parallels the single interest set argument in \Cref{thm:fixed-T}, so we streamline the exposition and focus on the new notation and technical details.

Let \(\{T_1,\dots,T_{\ell}\}\) be a family of pairwise disjoint interest sets, where \(T_i\subseteq [n]\) for each \(i\in[\ell]\). For each \(i\in[\ell]\) and each weight vector \(w^i\in \RR_+^{T_i}\), define \(\varphi_{w^i}: 2^{T_i}\rightarrow \RR_+\) by
\[
\varphi_{w^i}(S) :=
\begin{cases}
\min_{j \in S} w^i_j, & \text{if } S \neq \emptyset, \\
0, & \text{otherwise}.
\end{cases}
\qquad \text{for all } S \subseteq T_i.
\]

The Hall-type condition for the family \(\{T_1,\dots,T_\ell\}\)  can be obtained by  proving that there exists \(\lambda\in \Pi(X)\) such that
\begin{align}\label{eq:multi-T}
\sum_{i=1}^\ell \EE_{\mu\sim\lambda}\left[\varphi_{w^i}(K_\mu \cap T_i)\right]
\le
\sum_{i=1}^\ell \frac{1}{|T_i|} \sum_{j \in T_i} w^i_j,
\quad \text{for every collection } (w^1,\dots,w^\ell) \in \prod_{i=1}^\ell \RR_+^{T_i}.
\end{align}

To see this, fix any \(i\in[\ell]\) and any subset \(J\subseteq T_i\). Let \(\widetilde{w}^i = \mathbf{1}_J\), and let \(\widetilde{w}^k=\mathbf{0}\) for every \(k\neq i\). Then, for every \(S \subseteq T_i\),
\[
\varphi_{\widetilde{w}^i}(S)=
\begin{cases}
1, & \text{if } S \neq \emptyset \text{ and } S \subseteq J,\\
0, & \text{otherwise}.
\end{cases}
\]
Hence,
\begin{align*}
 \mathbb{P}_{\mu\sim\lambda}\!\left(K_{\mu}\cap T_i\neq\emptyset \text{ and } K_{\mu}\cap T_i\subseteq J\right) &= \mathbb{E}_{\mu\sim\lambda}\!\left[\varphi_{\widetilde{w}^i}(K_\mu \cap T_i)\right] \\
&\le \frac{1}{|T_i|}\sum_{j\in T_i} \widetilde{w}^i_j \\
&= \frac{|J|}{|T_i|}
\end{align*}
where the inequality follows from \Cref{eq:multi-T} applied to the collection \((\widetilde{w}^1,\dots,\widetilde{w}^\ell)\). Since $i\in [\ell]$ and \(J \subseteq T_i\) was arbitrary, this establishes statement of \Cref{thm:multi-T}. Thus, it remains to prove the existence of a distribution \(\lambda \in \Pi(X)\) satisfying \Cref{eq:multi-T}.

For \(\lambda \in \Pi(X)\) and \(w=(w^1,\dots,w^\ell) \in \prod_{i=1}^\ell \RR_+^{T_i}\), define
\[
G(\lambda,w) := \sum_{i\in [\ell]} \left( \EE_{\mu\sim\lambda} [ \varphi_{w^i} (K_\mu \cap T_i)  ] - \frac{1}{|T_i|} \sum_{j\in T_i} w_j^i \right).
\]
Then \Cref{eq:multi-T} is equivalent to the statement that \(\min_{\lambda}\max_w G(\lambda,w)\le 0\).

For each \(i\in[\ell]\), let
\[
g^{i}(\lambda,w^i):= \EE_{\mu\sim\lambda} [ \varphi_{w^i} (K_\mu \cap T_i)  ] - \frac{1}{|T_i|} \sum_{j\in T_i} w_j^i.
\]
Then \(G(\lambda,w) = \sum_{i\in[\ell]} g^{i}(\lambda,w^i)\). Moreover, for every \(c\ge 0\),
\[
g^{i}(\lambda, c w^i)= c\cdot g^{i}(\lambda, w^i).
\]
Thus each summand is positively homogeneous of degree one, and it is enough to restrict attention to the compact convex domain
\[
\Delta\coloneqq \prod_{i=1}^{\ell} \Delta(T_i),
\qquad
\Delta (T_i)\coloneqq \left\{\,w^{i}\in \bbR_+^{T_i} :
\sum_{j\in T_i} w_j^{i}\le 1\,\right\}.
\]

By the same reasoning as in \Cref{prop:key-T}, for each \(i\in[\ell]\) and each fixed \(\lambda_0 \in \Pi(X)\), the map \(w^i \mapsto g^{i}(\lambda_0,w^i)\) is continuous and concave. Likewise, for each fixed \(w_0^i\), the map \(\lambda \mapsto g^{i}(\lambda,w_0^i)\) is linear. Consequently, for fixed \(w_0=(w_0^1,\dots,w_0^\ell)\), the map \(\lambda \mapsto G(\lambda, w_0)\) is linear; and for fixed \(\lambda_0\), the map \(w\mapsto G(\lambda_0,w)\) is concave because each \(g^{i}(\lambda_0,w^i)\) is concave in \(w^i\). Therefore, \(G(\lambda,w):\Pi(X)\times \Delta \rightarrow \RR\) satisfies the hypotheses of Sion's minimax theorem. It follows that, to establish \Cref{eq:multi-T}, it is enough to prove
\[
\max_{w\in \Delta} \min_{\lambda \in \Pi(X)} G(\lambda,w) \leq 0.
\]
Equivalently, for every \(w\in \Delta\), there exists a decomposition \(\lambda_w\) such that \(G(\lambda_w,w)\le 0\). We prove a slightly stronger statement in the following lemma, which then implies \Cref{eq:multi-T} and hence \Cref{thm:multi-T}.

\begin{lemma}\label{lem:MultiT}
For each fixed $w=(w^1,...,w^\ell) \in \Delta$, there exists $\lambda_w \in \Pi(X)$ such that
\[
g^{i}(\lambda_w,w^i)\leq 0 \quad \text{ for every } i\in [\ell].
\]
\end{lemma}
\begin{proof}
Fix any $w\in \Delta$, write \(\|w^i \|_1 =\sum_{j\in T_i} w_j^{i}\), and partition the interest sets into \textit{hard} and \textit{easy} families:
\[
H=\left\{ i\in [\ell] : |T_i| \frac{r}{n} > 1 \text{ and } \|w^i \|_1 >0 \right\}, \qquad E=[\ell] \setminus H.
\]
\noindent\textbf{\textit{Easy interest sets.}} We first show that for every \(i\in E\), any \(\lambda \in \Pi(X)\) already satisfies \(g^{i}(\lambda,w^i)\le 0\). If \(\|w^i \|_1 =0\), then \(g^{i}(\lambda,w^i)=0\) trivially. Otherwise, since \(i\in E\), we must have \(|T_i| \frac{r}{n} \leq 1\). Therefore,
\begin{align*}
    \EE_{\mu\sim\lambda}[\varphi_{w^i}(K_\mu \cap T_i)] &\leq  \EE_{\mu\sim\lambda}\left[\sum_{j\in T_i} w^i_j \ \mathbb{I}[j\in K_\mu]\right] \\
    &= \sum_{j\in T_i} w^i_j\,\PP_{\mu\sim\lambda}(j\in K_{\mu}) \\
    &= \frac{r}{n} \sum_{j\in T_i} w^i_j \\
    &\leq \frac{1}{|T_i|} \sum_{j\in T_i} w^i_j.
\end{align*}
Here, the second equality holds because each row of \(X\) sums to exactly \(\frac{r}{n}\), and \(\lambda\) implements \(X\); therefore, for every \(j\), the probability that row \(j\) is matched is  \(\frac{r}{n}\).

Hence, for each $i\in E$, we have $ g^{i}(\lambda,w^i)\le 0$ for all $\lambda \in \Pi(X)$. It therefore suffices to construct \(\lambda_w\in \Pi(X)\) such that \(g^j(\lambda_w,w^j)\le 0\) for all \(j\in H\).

\medskip

\noindent\textbf{\textit{Hard interest sets.}}
For each $i\in H$, relabel the indices of $T_i$ so that
\begin{align}\label{index-mT}
    w^i_{1} \ge w^i_{2} \ge \cdots \ge w^i_{|T_i|}.
\end{align}
where these indices are understood locally within \(T_i\). This local indexing is well defined because the sets \(T_i\) are pairwise disjoint.

For each $i\in H$, and each \(k\in \{1,\dots, |T_i|\}\), define \(\delta^i_k := w^i_k - w^i_{k+1}\), with the convention that \(w^i_{|T_i| +1}=0\), and set \(J^i_k := \{1,\dots,k\} \subseteq T_i\). With this notation, the same argument as in \Cref{lem:Single-T} shows that every nonempty subset \(S \subseteq T_i\) satisfies
\begin{align}\label{eq:phi-Identity-multi}
    \varphi_{w^i}(S)=\sum_{k=1}^{|T_i|} \delta^i_k\, \mathbb{I}[S\subseteq J^i_k].
\end{align}
Substituting \Cref{eq:phi-Identity-multi} into the expectation yields, for any \(\lambda\in \Pi(X)\),
\begin{align}\label{id:exp-eq-multi-T}
    \EE_{\mu\sim\lambda}[\varphi_{w^i}(K_\mu \cap T_i)] &=   \sum_{k=1}^{|T_i|} \delta^i_k \ \PP_{\mu\sim\lambda} \!\left(K_{\mu}\cap T_i\neq\emptyset \text{ and } K_{\mu}\cap T_i\subseteq J^i_k\right).
\end{align}
Thus, to find a single \(\lambda_w \in \Pi(X)\) such that \(g^i(\lambda_w, w^i)\le 0\) for every \(i\in H\), it is enough to prove that, for each \(i\in H\),
\begin{align}\label{ineq:multi-T}
    \PP_{\mu\sim\lambda_w} \!\left(K_{\mu}\cap T_i\neq\emptyset \text{ and } K_{\mu}\cap T_i\subseteq J^i_k\right) \leq \frac{k}{|T_i|} \qquad \text{ for each } k\in \{1,...,|T_i|\}.
\end{align}
Indeed, once \Cref{ineq:multi-T} is established, the lemma follows by substituting this bound into \Cref{id:exp-eq-multi-T}:
 \begin{align*}
    \EE_{\mu\sim\lambda_w}[\varphi_{w^i}(K_\mu \cap T_i)] &\leq \sum_{k=1}^{|T_i|} \delta^i_k \frac{k}{|T_i|}
    = \frac{1}{|T_i|} \sum_{k=1}^{|T_i|} k \ \delta^i_k
    = \frac{1}{|T_i|} \sum_{k=1}^{|T_i|} w^i_k \quad \text{ for each }i\in H.
\end{align*}
Hence \(g^i(\lambda_w,w^i)\le 0\) for all \(i\in H\), as required.

We have  reduced the lemma to constructing a distribution $\lambda_w$ that satisfies the family of bounds given in \Cref{ineq:multi-T}. We obtain such decomposition $\lambda_w$ of $X$ by a carefully constructed flow network with a feasible flow and subsequently applying \Cref{lem:flow}.

\begin{figure}[t]
    \centering
\begin{tikzpicture}[
    >=Stealth,
    every node/.style={font=\sffamily},
    scale=0.88,
    transform shape,
    term/.style={
        circle, draw=blue!55!black, fill=blue!10,
        line width=0.9pt, minimum size=18pt,
        font=\sffamily\scriptsize\bfseries
    },
    gadget/.style={
        circle, draw=teal!60!black, fill=teal!10,
        line width=0.8pt, minimum size=15pt
    },
    rownode/.style={
        circle, draw=orange!75!black, fill=orange!12,
        line width=0.8pt, minimum size=15pt
    },
    itemnode/.style={
        rounded corners=3pt, draw=violet!55!black, fill=violet!10,
        line width=0.8pt, minimum width=15pt, minimum height=15pt
    },
    arr/.style={->, line width=0.8pt, black!65},
    dots/.style={font=\Large, text=black!55},
    groupbox/.style={
        rounded corners=6pt, draw=teal!55!black, fill=teal!4,
        line width=0.9pt, dashed, inner xsep=8pt, inner ysep=7pt
    },
    collabel/.style={font=\sffamily\footnotesize\bfseries, text=black!65},
    grouplabel/.style={
        font=\sffamily\scriptsize\bfseries,
        text=teal!60!black,
        fill=white,
        inner sep=1.5pt,
        rounded corners=2pt
    }
]

\newcommand{\vstrongdots}{\(\vdots\)}

\def\srcX{-0.9}
\def\chainX{2.2}
\def\rowX{4.2}
\pgfmathsetmacro{\midX}{(\chainX+\rowX)/2}
\def\itemX{7.2}
\def\sinkX{9.4}

\node[collabel] at (\chainX, 1.35) {Chain nodes};
\node[collabel] at (\rowX, 1.35) {Rows};
\node[collabel] at (\itemX, 1.35) {Items};

\node[term] (s) at (\srcX, -7.5) {$s$};

\node[gadget] (c1a) at (\chainX, 0) {};
\node[gadget] (c1b) at (\chainX, -1.0) {};
\node[dots]    at (\chainX, -1.75) {\vstrongdots};
\node[gadget] (c1c) at (\chainX, -2.35) {};

\draw[arr] (c1a) -- (c1b);
\draw[arr] (c1b) -- ++(0, -0.55);

\node[rownode] (r1a) at (\rowX, 0) {};
\node[rownode] (r1b) at (\rowX, -1.0) {};
\node[dots]    at (\rowX, -1.75) {\vstrongdots};
\node[rownode] (r1c) at (\rowX, -2.35) {};

\draw[arr] (c1a) -- (r1a);
\draw[arr] (c1b) -- (r1b);
\draw[arr] (c1c) -- (r1c);
\draw[arr] (s) to[out=55, in=180] (c1a);

\begin{scope}[on background layer]
    \node[groupbox, fit=(c1a)(c1b)(c1c)(r1a)(r1b)(r1c)] (TOneBox) {};
\end{scope}
\node[grouplabel, anchor=south west] at ($(TOneBox.north west)+(0.08,0.05)$) {$T_1$};

\def\gTwoY{-3.9}

\node[gadget] (c2a) at (\chainX, \gTwoY) {};
\node[gadget] (c2b) at (\chainX, \gTwoY - 1.0) {};
\node[dots]    at (\chainX, \gTwoY - 1.75) {\vstrongdots};
\node[gadget] (c2c) at (\chainX, \gTwoY - 2.35) {};

\draw[arr] (c2a) -- (c2b);
\draw[arr] (c2b) -- ++(0, -0.55);

\node[rownode] (r2a) at (\rowX, \gTwoY) {};
\node[rownode] (r2b) at (\rowX, \gTwoY - 1.0) {};
\node[dots]    at (\rowX, \gTwoY - 1.75) {\vstrongdots};
\node[rownode] (r2c) at (\rowX, \gTwoY - 2.35) {};

\draw[arr] (c2a) -- (r2a);
\draw[arr] (c2b) -- (r2b);
\draw[arr] (c2c) -- (r2c);
\draw[arr] (s) to[out=28, in=180] (c2a);

\begin{scope}[on background layer]
    \node[groupbox, fit=(c2a)(c2b)(c2c)(r2a)(r2b)(r2c)] (TTwoBox) {};
\end{scope}
\node[grouplabel, anchor=south west] at ($(TTwoBox.north west)+(0.08,0.05)$) {$T_2$};

\node[dots] at (\chainX, -7.2) {\vstrongdots};
\node[dots] at (\rowX,   -7.2) {\vstrongdots};
\node[dots] at (\midX,  -7.2) {\vstrongdots};

\def\gEllY{-8.6}

\node[gadget] (cla) at (\chainX, \gEllY) {};
\node[gadget] (clb) at (\chainX, \gEllY - 1.0) {};
\node[dots]    at (\chainX, \gEllY - 1.75) {\vstrongdots};
\node[gadget] (clc) at (\chainX, \gEllY - 2.35) {};

\draw[arr] (cla) -- (clb);
\draw[arr] (clb) -- ++(0, -0.55);

\node[rownode] (rla) at (\rowX, \gEllY) {};
\node[rownode] (rlb) at (\rowX, \gEllY - 1.0) {};
\node[dots]    at (\rowX, \gEllY - 1.75) {\vstrongdots};
\node[rownode] (rlc) at (\rowX, \gEllY - 2.35) {};

\draw[arr] (cla) -- (rla);
\draw[arr] (clb) -- (rlb);
\draw[arr] (clc) -- (rlc);
\draw[arr] (s) to[out=-28, in=180] (cla);

\begin{scope}[on background layer]
    \node[groupbox, fit=(cla)(clb)(clc)(rla)(rlb)(rlc)] (Tellbox) {};
\end{scope}
\node[grouplabel, anchor=south west] at ($(Tellbox.north west)+(0.08,0.05)$) {$T_\ell$};

\node[rownode] (rn) at (\rowX, -11.9) {};
\node[dots] (moremid) at ($(rlc)!0.5!(rn)$) {\vstrongdots};

\draw[arr] (s) .. controls +(1.3,-3.8) and +(-1.0,-0.2) .. (rn);

\node[itemnode] (i1)   at (\itemX, 0) {};
\node[itemnode] (i2)   at (\itemX, -1.8) {};
\node[dots]      at (\itemX, -3.9) {\vstrongdots};
\node[itemnode] (imm2) at (\itemX, -6.3) {};
\node[itemnode] (imm1) at (\itemX, -8.4) {};
\node[dots]      at (\itemX, -10.2) {\vstrongdots};
\node[itemnode] (im)   at (\itemX, -11.8) {};

\draw[arr] (r1a) -- (i1);
\draw[arr] (r1b) -- (i2);
\draw[arr] (r1c) -- (i2);
\draw[arr] (r2a) -- (i1);
\draw[arr] (r2b) -- (imm2);
\draw[arr] (r2c) -- (i2);
\draw[arr] (rla) -- (imm1);
\draw[arr] (rlb) -- (imm1);
\draw[arr] (rlc) -- (imm1);
\draw[arr] (rn) -- (imm2.west);
\draw[arr] (rn) -- (im.west);

\node[term] (t) at (\sinkX, -5.0) {$t$};

\draw[arr] (i1) -- (t);
\draw[arr] (i2) -- (t);
\draw[arr] (imm2) -- (t);
\draw[arr] (imm1) -- (t);
\draw[arr] (im) -- (t);

\end{tikzpicture}
    \caption{Flow network used in the proof of \Cref{thm:multi-T}. Each dashed box represents the gadget associated with one interest set \(T_i\).}
    \label{fig:multi-T}
\end{figure}
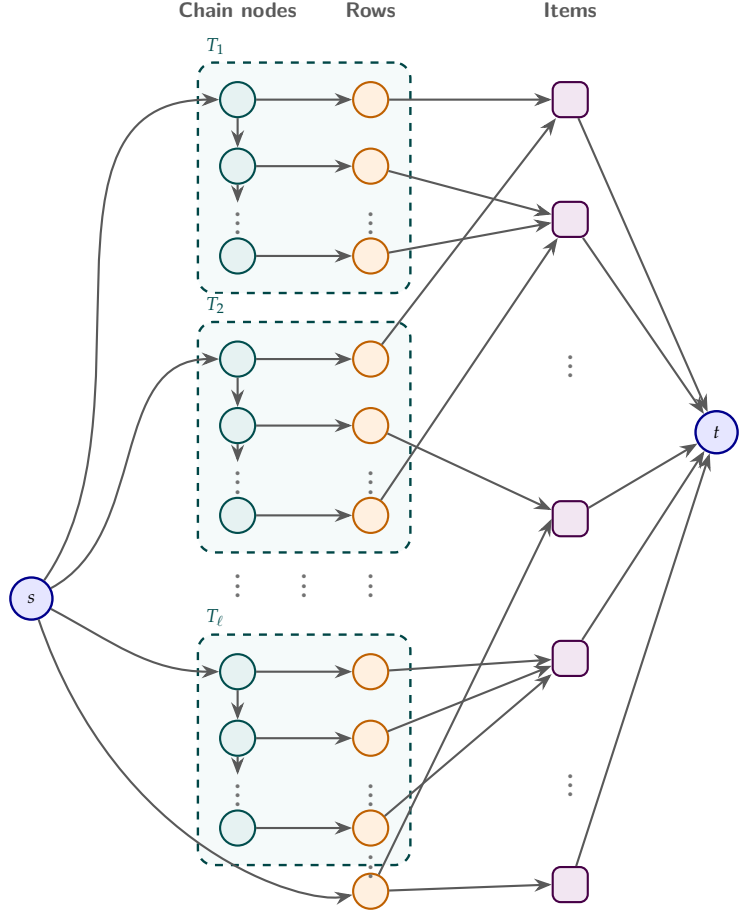

\medskip
\noindent\textit{Biased Flow Network for Multiple Interest Sets.} We construct a directed flow network as shown in \Cref{fig:multi-T}. Formally, we build a directed network with source $s$, sink $t$, row nodes $i$ for each $i \in [n]$, item nodes $c_j$ for each $j \in [m]$, and, for each $q \in H$, a chain of nodes $\tau^q_1, \dots, \tau^q_{|T_q|}$.
The arcs are given as follows:
\begin{itemize}[leftmargin=2em]
\item for each $q\in H$:
\begin{itemize}
    \item[$\scriptstyle\blacktriangleright$] $s\to \tau^q_1$ with lower and upper capacities $\ell = \lfloor \frac{r}{n} |T_q| \rfloor $ and $u= \lceil \frac{r}{n} |T_q| \rceil$;
    \item[$\scriptstyle\blacktriangleright$] for each $k\in \{1,\dots,|T_q|-1\}$, an arc $\tau^q_k\to \tau^q_{k+1}$ with lower and upper capacities $\ell = \lfloor \frac{r}{n} (|T_q|-k) \rfloor $ and $u= \lceil \frac{r}{n} (|T_q|-k) \rceil$ ;
    \item[$\scriptstyle\blacktriangleright$] for each $k\in \{1,\dots,|T_q|\}$, an arc $\tau^q_k\to k$ with lower and upper capacities 0 and 1, where $k$ here denotes the row node in $T_q$;
\end{itemize}
\item for each row $i\notin \bigcup_{q\in H} T_q$, an arc $s\to i$ with lower and upper capacities 0 and 1;
\item for each pair $(i,j)$ with $X_{ij}>0$, an arc $i\to c_j$ with lower and upper capacities 0 and 1;
\item for each item node $c_j$, an arc $c_j\to t$ with lower and upper capacities 0 and 1.
\end{itemize}
Note that lower and upper bounds on the flows on each arc are integral.
We now define a feasible fractional $s$-$t$ flow $f$ on this network with value $r$:
\begin{itemize}
    \item  for each $q\in H$:
    \begin{itemize}
        \item[$\scriptstyle\blacktriangleright$] $f(s,\tau^q_1)=\frac{r}{n} |T_q| $;
        \item[$\scriptstyle\blacktriangleright$] for each $k\in \{1,\dots,|T_q|-1\}$, set $f(\tau^q_k,\tau^q_{k+1})=\frac{r}{n}(|T_q|-k)$;
        \item[$\scriptstyle\blacktriangleright$] for each $k\in \{1,\dots,|T_q|\}$, set $f(\tau^q_k,k)=\frac{r}{n}$.
        Recall that the indexing within each $T_q$ is defined locally (see~\eqref{index-mT}), with index $k$ denoting the row node in $T_q$ having the $k$th highest weight.
    \end{itemize}
    \item  for each $i\notin \bigcup_{q\in H} T_q$, set $f(s,i)=\frac{r}{n}$;
    \item for each arc $(i,c_j)$, set $f(i,c_j)=X_{ij}$;
    \item for each arc $(c_j,t)$, set $f(c_j,t)=\sum_{i=1}^n X_{ij}$.
\end{itemize}

Note that the flow $f$ is feasible since $0\leq \frac{r}{n}\leq 1$, and since each column sum of $X$ is at most 1, we have $0\leq \sum_{i=1}^n X_{ij}\leq 1$. Flow conservation at row nodes follows from the fact that each row of \(X\) sums to exactly \(\frac{r}{n}\), and flow conservation at the chain nodes is immediate from the definition of \(f\).
The total flow value is $\sum_{i\in [n]}\sum_{j\in [m]} X_{ij}=r$. Thus, the network and the flow satisfy the conditions of \Cref{lem:flow}, hence can be expressed as a distribution $\lambda_w$ over integral feasible flows of value $r$.

Fix any integral feasible flow in this decomposition. This flow defines a matching of size $r$, since every row node has inflow at most $1$ and therefore sends flow along at most one row--item arc, while every item node has capacity $1$ on its outgoing arc to the sink and hence is incident to at most one matched row. Hence, the distribution $\lambda_w$ over the integral feasible flows of value $r$ also gives distribution over size $r$ matchings.

Moreover, this distribution over matchings implements \(X\) i.e., \(\lambda_w\in \Pi(X)\). Indeed, for every pair \((i,j)\) with \(X_{ij}>0\), the edge \((i,j)\) is included in the matching with probability \(f(i,c_j)=X_{ij}\), since the integral flows decompose the fractional flow.

\medskip

Now fix any $i\in H$, let $R^i_{k}:=\{k,...,|T_i|\}$ be the suffix of $T_i$ starting at index $k$. For every \(\mu\) in the support of \(\lambda_w\), the set of matched rows \(K_\mu\) satisfies
\begin{align}\label{eq:suffix-multi-T}
    |K_\mu \cap R^i_{k+1}| \in \left\{   \left\lfloor \frac{r}{n} (|T_i|-k) \right\rfloor , \left\lceil \frac{r}{n} (|T_i|-k) \right\rceil \right\} \quad \text{ for each } k\in \{0,...,|T_i|-1\}.
\end{align}
To see this, first consider \(k\ge 1\). The arc \(\tau^i_k\to \tau^i_{k+1}\) has lower and upper capacities \(\lfloor \frac{r}{n} (|T_i|-k) \rfloor\) and \(\lceil \frac{r}{n} (|T_i|-k) \rceil\). Since the flow corresponding to \(\mu\) is integral, the amount of flow on this arc must be an integer between these two bounds. By flow conservation, this quantity is exactly \(|K_\mu \cap R^i_{k+1}|\), because each matched row in \(R^i_{k+1}\) contributes one unit of flow and each row can be matched to at most one item node.
When \(k=0\), the arc \(s\to \tau^i_1\) has lower and upper capacities \(\lfloor \frac{r}{n} |T_i| \rfloor\) and \(\lceil \frac{r}{n} |T_i| \rceil\), and the same argument shows that its flow value is exactly \(|K_\mu \cap R^i_{1}|=|K_\mu \cap T_i|\). This proves \Cref{eq:suffix-multi-T}.

Since $i\in H$, we have \(\frac{r|T_i|}{n} > 1\). Hence  \eqref{eq:suffix-multi-T} implies that $|K_\mu \cap T_i|\geq\lfloor \frac{r}{n} |T_i| \rfloor  \geq 1$. Thus, every matching $\mu$ in the support of $\lambda_w$ satisfies $K_{\mu}\cap T_i \neq \emptyset$. It follows that, for each $k\in \{1,...,|T_i|-1\}$, the following events are equivalent:
\[K_{\mu}\cap T_i\neq\emptyset \text{ and } K_{\mu}\cap T\subseteq J^i_k \iff |K_{\mu} \cap R^i_{k+1}|=0. \]

\noindent\textit{Verifying $\lambda_w$ satisfies \Cref{ineq:multi-T}.}
Recall that \Cref{ineq:multi-T} imposes, for each hard interest set, a family of inequalities. For the fixed \(i\in H\), it requires
\begin{align}\label{eq:mTT}
    \PP_{\mu\sim\lambda_w} \!\left(K_{\mu}\cap T_i\neq\emptyset \text{ and } K_{\mu}\cap T_i\subseteq J^i_k\right) \leq \frac{k}{|T_i|} \qquad \text{ for each } k\in \{1,...,|T_i|\}.
\end{align}
Note that the inequality \eqref{eq:mTT} is trivial for \(k=|T_i|\), so we may assume that \(k\in\{1,\dots,|T_i|-1\}\). Define
\[
Y^i_{k+1}:=|K_\mu\cap R^i_{k+1}|.
\]
By identity \eqref{eq:suffix-multi-T}, this random variable takes values in $\left\{   \left\lfloor \frac{r}{n} (|T_i|-k) \right\rfloor , \left\lceil \frac{r}{n} (|T_i|-k) \right\rceil \right\}$.
We now verify inequality~\eqref{eq:mTT} by considering the following two cases.

\medskip

\noindent\textit{Case 1: $\frac{r}{n} (|T_i|-k)\geq 1$.} In this case, we have that $Y^i_{k+1} = |K_\mu\cap R^i_{k+1}| \geq \lfloor \frac{r}{n} (|T_i|-k) \rfloor\geq 1$ for any matching $\mu$ in the support of $\lambda_w$. Hence,
\[
\PP_{\mu\sim\lambda_w} \!\left(K_{\mu}\cap T_i\neq\emptyset \text{ and } K_{\mu}\cap T_i\subseteq J^i_k\right) = \PP_{\mu\sim\lambda_w} \!\left( Y^i_{k+1}=0\right)=0.
\]

\medskip

\noindent\textit{Case 2:  $\frac{r}{n} (|T_i|-k)< 1$.}
Since \(\lambda_w\) implements \(X\), each row \(i\) is matched with probability exactly \(\frac{r}{n}\). As \(|R^i_{k+1}|=|T_i|-k\), it follows by linearity of expectation that
\(
\EE_{\mu\sim\lambda_w}[Y^i_{k+1}]
=
\EE_{\mu\sim\lambda_w}[|K_\mu\cap R^i_{k+1}|]
=
\frac{r}{n}(|T_i|-k).
\)
Furthermore, by identity \eqref{eq:suffix-multi-T} and the case distinction, we know that the random variable $Y^i_{k+1}$ takes values in $\{0,1\}$. Hence, we have $\PP_{\mu\sim\lambda_w} \!\left( Y^i_{k+1} =1 \right) = \EE_{\mu\sim\lambda_w}[|K_\mu\cap R^i_{k+1}|] =  \frac{r}{n}(|T_i|-k)$. We obtain,
\begin{align*}
&\PP_{\mu\sim\lambda_w}
\!\left(K_{\mu}\cap T_i\neq\emptyset
\text{ and }K_{\mu}\cap T_i\subseteq J^i_k\right)\\
&\qquad = \PP_{\mu\sim\lambda_w}\!\left(Y^i_{k+1}=0\right)
= 1 - \frac{r}{n}(|T_i|-k)
< 1- \frac{|T_i|-k}{|T_i|}
=\frac{k}{|T_i|}.
\end{align*}
Here, the inequality follows since $i\in H$, we know that \(\frac{r|T_i|}{n} > 1\).

Thus inequality~\eqref{eq:mTT} holds for the fixed $i\in H$.
Since \(i \in H\) was arbitrary, the same conclusion holds for every hard interest set. This completes the proof that \(\lambda_w\) satisfies the family of inequalities given in \Cref{ineq:multi-T}, and thus establishes the lemma.

\end{proof}

\section{Stochastic Dominance Fairness}
\label{sect:SD-BoBW}
In this section, we consider the setting where agents only have ordinal preferences, and consider the corresponding fairness notions, stochastic dominance envy-freeness (SD-EF) and stochastic dominance EF$1$.
We show that the PS-lottery by~\citet{AzizFrSh24} can be extended to guarantee \exan SD-EF and \expo SD-EF$2$ (Theorem~\ref{thm:SD-EF+SD-EF2}).
However, in contrast to the main result of this paper, we show that \exan SD-EF is incompatible with \expo SD-EF$1$ (Theorem~\ref{thm:neg:SD-EF+SD-EF1-goods-chores}).

For an agent $i\in N$ and any pair of items $h,h'\in M$, denote by $h\succsim_i h'$ (resp., $h\succ_i h'$) if $i$ weakly prefers (resp., strictly prefers) $h$ over $h'$.
For each agent $i$, let $G_i$ be the set of items that $i$ weakly prefers over a null item, i.e., $G_i:=\{h\in M:h\succsim_i\emptyset\}$.
Let $Z_i$ be the set of items such that $i$ strictly prefers a null item over it, i.e., $Z_i:=\{h\in M:h\prec_i \emptyset\}$.

An item $g$ is called a \emph{subjective good} if there exists an agent $i$ such that $g\in G_i$.
An item $c$ is called an \emph{objective chore} if for all agents $i$, $c\in Z_i$.
As before, we still partition the item set $M$ into a set of subjective goods $G$ (in which each item is denoted by $g$) and a set of objective chores $Z$ (in which each item is denoted by $c$).
Note that $G=\bigcup_{i\in N}G_i$, and $Z=\bigcap_{i\in N}Z_i$.

Envy-freeness and EF$1$ naturally extend to the stochastic dominance envy-freeness (SD-EF) and stochastic dominance EF$1$ under the ordinal setting.

\begin{definition}\label{def:sd-ef}
    A fractional allocation~$X$ satisfies \emph{stochastic dominance envy-freeness (\SDEF)} if for any pair of agents~$i, j \in N$, we have $X_i \succsim_i^{\text{SD}} X_j$, i.e., for each item $g\succ_i \emptyset$,
    we have
    \[
    \sum_{g' \in G \colon g' \succsim_i g} X_{i g'} \geq \sum_{g' \in G \colon g' \succsim_i g} X_{j g'},
    \]
    and for each item $c\prec_i \emptyset$, we have
    \[
    \sum_{c' \in Z \colon c' \precsim_i c} X_{i c'} \leq \sum_{c' \in Z \colon c' \precsim_i c} X_{j c'}.
    \]
\end{definition}

\begin{definition}\label{def:sd-ef1}
    An integral allocation~$A$ satisfies \emph{stochastic dominance EF$1$ (\SDEFOne)} if for any pair of agents~$i, j \in N$, either of the following holds:
    \begin{itemize}
        \item $i$ does not SD-envy $j$: $A_i\succsim^{\text{SD}}_i A_j$, or
        \item there exists $h\in A_i\cup A_j$ such that $A_i\setminus\{h\}\succsim^{\text{SD}}_i A_j\setminus\{h\}$.
    \end{itemize}
    Alternatively, an integral allocation $A$ is \SDEFOne if for any pair of agents~$i, j \in N$, either of the following holds:
    \begin{itemize}
        \item $A_i \succsim_i^{\text{SD}} A_j \setminus \{h\}$ for some~$h \in A_j$ with $h\succsim_i\emptyset$, or
        \item $A_i \setminus \{h\} \succsim_i^{\text{SD}} A_j$ for some~$h \in A_i$ with $\emptyset\succ_ih$.
    \end{itemize}
\end{definition}

Given an integral allocation~$A$ and a subset of items~$S \subseteq M$, let $A|_{S} \coloneqq (A_1|_{S}, A_2|_{S}, \dots, A_n|_{S})$ be the allocation restricted to items~$S$, where for each~$i \in N$, $A_i|_{S} \coloneqq A_i \cap S$.

\SDEFOne allocations satisfy the following property.
\begin{lemma}
\label{lemma:no-simultaneous-SD-envy}
Consider any \SDEFOne allocation~$A$.
For all~$i, j \in N$, $A_j|_{G_i} \succ_i^{\text{SD}} A_i|_{G_i}$ and $A_j|_{Z_i} \succ_i^{\text{SD}} A_i|_{Z_i}$ cannot happen simultaneously.
\end{lemma}
\begin{proof}
Suppose for the sake of contradiction that there exists an \SDEFOne allocation~$A$ of the instance such that $A_j |_{G_i} \succ_i^{\text{SD}} A_i |_{G_i}$ and $A_j |_{Z_i} \succ_i^{\text{SD}} A_i |_{Z_i}$.
Since allocation~$A$ is \SDEFOne, there exists an item~$h \in A_i \cup A_j$ such that $A_i \setminus \{h\} \succsim_i^{\text{SD}} A_j \setminus \{h\}$.
However, if $h \in G_i$, we still have $A_j|_{Z_i} \succ_i^{\text{SD}} A_i|_{Z_i}$, contradicting the fact that $A$ is \SDEFOne.
A similar argument holds when $h \in Z_i$.
\end{proof}

We further define a relaxation of SD-EF$1$, which allows envy up to the removal of two items.

\begin{definition}\label{def:sd-ef2}
    An integral allocation~$A$ satisfies \emph{stochastic dominance EF$2$ (\SDEFTwo)} if for any pair of agents~$i, j \in N$, there exist two items from $A_i\cup A_j$ such that $i$ does not SD-envy $j$ after the removal of the two items from their corresponding bundles, i.e.,
    there exist $h,h'\in A_i\cup A_j$ such that $A_i\setminus\{h,h'\}\succsim^{\text{SD}}_i A_j\setminus\{h,h'\}$.
\end{definition}

Under the ordinal preferences, known approaches achieve \exan \SDEF and \expo \SDEFTwo.
\begin{theorem}\label{thm:SD-EF+SD-EF2}
    Under the mixed goods and chores setting with ordinal preferences, \exan \SDEF and \expo \SDEFTwo are compatible.
\end{theorem}
\begin{proof}
    It is known that PS-lottery achieves \exan \SDEF and \expo \SDEFOne for objective goods \cite{AzizFrSh24}, and can be similarly extended to the setting that contains only subjective goods or objective chores as in Section~\ref{sect:PS}.
    We can apply PS-lotteries on the set of subjective goods $G$ (with a sufficient number of dummy goods added to ensure no one receives any chore), and the set of objective chores $Z$ separately.
    We then arbitrarily combine the two obtained lotteries while preserving the marginal probabilities.

    Ex-ante \SDEF of the combined lottery directly follows from \exan \SDEF of the lotteries on $G$ and $Z$.
    Moreover, since in each integral allocation in the lotteries, at most one item needs to be removed on $G$ and at most one on $Z$ (as each lottery is \expo \SDEFOne), combining them leads to at most two items' removal in total, thus \expo \SDEFTwo.
\end{proof}

Different from the cardinal setting where \exan EF and \expo EF$1$ are compatible, we show in the following that \exan \SDEF and \expo \SDEFOne cannot be simultaneously achieved.
Therefore, \expo \SDEFTwo is the best we can guarantee when restricting to \exan \SDEF.

\begin{theorem}
\label{thm:neg:SD-EF+SD-EF1-goods-chores}
There exists an instance with three agents and mixed goods and chores in which no randomized allocation is simultaneously \exan \SDEF and \expo \SDEFOne.
\end{theorem}

\begin{proof}
Consider an instance with agents $N = \{1, 2, 3\}$ and items $G \cup Z$ where $G = \{g_1, g_2, g_3\}$ and $Z = \{c_1, c_2, c_3\}$.
The agents have the following ordinal preferences over the items:
\begin{itemize}
\item Agent~$1$: $g_1 \succ g_2 \succ g_3 \succ \emptyset \succ c_1 \succ c_2 \succ c_3$;
\item Agent~$2$: $g_1 \succ g_3 \succ g_2 \succ \emptyset \succ c_1 \succ c_2 \succ c_3$;
\item Agent~$3$: $g_2 \succ g_1 \succ g_3 \succ \emptyset \succ c_1 \succ c_2 \succ c_3$.
\end{itemize}
In other words, all agents consider items in~$G$ as (objective) goods and items in~$Z$ as (objective) chores, and they have an identical ordinal preference over the chores.

By the definition of \SDEFOne, in any \SDEFOne integral allocation, every agent gets exactly one item from~$G$ and exactly one item from~$Z$.
First, we claim that in any \SDEFOne integral allocation of the instance, agent~$2$ does not receive item~$g_2$.
Suppose for the sake of contradiction that there exists an \SDEFOne allocation~$A$ in which agent~$2$ gets item~$g_2$.
Since each agent gets exactly one item from~$G$ and item~$g_2$ is the least (resp., most) preferred item for agent~$2$ (resp., agent~$3$), we have
\[
A_3|_G \succ_2^\SD A_2|_G \qquad \text{and} \qquad A_2|_G \succ_3^\SD A_3|_G.
\]
By \Cref{lemma:no-simultaneous-SD-envy}, we must have
\[
A_3|_Z \precsim_2^\SD A_2|_Z \qquad \text{and} \qquad A_2|_Z \precsim_3^\SD A_3|_Z.
\]
Put differently, in allocation~$A$, both agent~$2$ and~$3$ prefer their own chore to the chore received by the other agent.
This is impossible as both agents have an identical ordinal preference over chores~$Z$.
Next, we claim that in any \SDEFOne integral allocation, agent~$3$ does not receive item~$g_1$.
Suppose for the sake of contradiction that in an \SDEFOne allocation~$A$, agent~$3$ gets item~$g_1$.
Since each agent gets exactly one item from~$G$ and agent~$2$ never gets item~$g_2$, we have $$A|_G = (\{g_2\}, \{g_3\}, \{g_1\}),$$ and thus
\[
A_3|_G \succ_1^\SD A_1|_G \qquad \text{and} \qquad A_1|_G \succ_3^\SD A_3|_G.
\]
Again, by \Cref{lemma:no-simultaneous-SD-envy}, we must have that in allocation~$A$, both agent~$1$ and~$3$ prefer their own chore to that of the other agent, and this is impossible.
To summarize, we have so far shown that in any \SDEFOne integral allocation of the instance, agent~$2$ does not receive item~$g_2$ and agent~$3$ does not receive item~$g_1$.
As a result, we are left with the following three possible integral allocations of items~$G$ in any \SDEFOne integral allocation of the instance:
\begin{enumerate}[label=(\roman*)]
\item $(\{g_1\}, \{g_3\}, \{g_2\})$;
\item $(\{g_3\}, \{g_1\}, \{g_2\})$;
\item $(\{g_2\}, \{g_1\}, \{g_3\})$.
\end{enumerate}

Suppose for the sake of contradiction that there exists an \exan \SDEF fractional allocation~$X$ that can be implemented by \expo \SDEFOne integral allocations.
First, following the definition of \SDEF, we must have $X_{1 g_1} \leq 1/2$ and $X_{2 g_1} \leq 1/2$, i.e., each of agents~$1$ and~$2$ must get at most~$\frac{1}{2}$ of item~$g_1$.
Moreover, since agent~$3$ never gets item~$g_1$ in \expo integral allocations, agent~$3$ must not receive any fraction of~$g_1$ \exan, i.e., $X_{3 g_1} = 0$.
It follows that
\[
X_{1 g_1} = X_{2 g_1} = 1/2.
\]
Similarly, as agent~$2$ never gets item~$g_2$ in any \expo integral allocation, we must have $X_{2 g_2} = 0$.
Now consider agents~$1$ and~$3$.
Due to \SDEF, we must satisfy the following two inequalities simultaneously:
\[
X_{1 g_1} + X_{1 g_2} \geq X_{3 g_1} + X_{3 g_2} \qquad \text{and} \qquad X_{3 g_2} + X_{3 g_1} \geq X_{1 g_2} + X_{1 g_1},
\]
which leads to $X_{1 g_2} = 1/4$ and $X_{3 g_2} = 3/4$.
Applying a similar argument to reason the \SDEF property between agents, we have
\[
X_{1 g_3} = X_{3 g_3} = 1/4 \qquad \text{and} \qquad X_{2 g_3} = 1/2.
\]
To summarize, we have so far reached the unique fractional allocation of items~$G$ that can be decomposed into the three possible integral allocations of items~$G$.
Furthermore, the implementation of~$X|_G$ is unique and as follows:
\[
\begin{bmatrix}
1/2 & 1/4 & 1/4 \\
1/2 & 0 & 1/2 \\
0 & 3/4 & 1/4
\end{bmatrix}
= 1/2 \cdot (\{g_1\}, \{g_3\}, \{g_2\}) + 1/4 \cdot (\{g_3\}, \{g_1\}, \{g_2\}) + 1/4 \cdot (\{g_2\}, \{g_1\}, \{g_3\}).
\]

From the above implementation, we have the following facts:
\begin{itemize}
\item With probability~$1/2$, agent~$1$ SD-envies agent~$2$.
\item With probability~$1/2$, agent~$2$ SD-envies agent~$1$.
\end{itemize}
It implies that when allocating items~$Z$, agent~$1$ should be SD-envy-free towards agent~$2$ with probability~$1/2$, and vice versa.
Among the six possible integral allocations of items~$Z$ between the three agents, agent~$1$ is SD-envy-free towards agent~$2$ in the following three allocations of~$Z$:
\begin{itemize}
\item Agent~$1$ gets~$c_1$ and allocates $c_2, c_3$ to agents~$2$ and~$3$ arbitrarily; and
\item $(\{c_2\}, \{c_3\}, \{c_1\})$.
\end{itemize}
Recall that fractional allocation~$X$ is \SDEF and all agents have the identical ordinal preference over items~$Z$, which implies that each agent gets exactly~$1/3$ of each item in~$Z$.
We thus conclude that the allocation $(\{c_2\}, \{c_3\}, \{c_1\})$ appears with probability~$1/2 - 1/3 = 1/6$.

Recall from the implementation, with probability~$1/4$, the integral allocation of items~$G$ is $(\{g_2\}, \{g_1\}, \{g_3\})$, in which agent~$1$ SD-envies agent~$2$ and agent~$3$ SD-envies both agents~$1$ and~$2$.
Then, conditioning on $(\{g_2\}, \{g_1\}, \{g_3\})$, in order to have an \SDEFOne allocation of items~$G \cup Z$, there is a unique integral allocation of items~$Z$: $(\{c_2\}, \{c_3\}, \{c_1\})$.
In other words, the allocation $(\{c_2\}, \{c_3\}, \{c_1\})$ should appear with probability~$1/4$, which contradicts the probability~$1/6$ in the last paragraph, as desired.
\end{proof}
\end{document}